\documentclass[12pt]{extarticle}
\usepackage[utf8]{inputenc}
\usepackage[T1]{fontenc}
\usepackage{graphicx}
\usepackage{framed}
\usepackage[normalem]{ulem}
\usepackage{amsmath}
\usepackage{amsthm}
\usepackage{amssymb}
\usepackage{amsfonts}
\usepackage{enumerate}
\usepackage[top=1 in,bottom=1in, left=1 in, right=1 in]{geometry}
\usepackage{subcaption}
\usepackage[usenames,dvipsnames]{xcolor}
\usepackage{enumitem}
\usepackage{authblk}
\usepackage[title]{appendix}
\usepackage{algorithm}
\usepackage{algpseudocode}
\usepackage{multirow}
\usepackage{makecell}

\usepackage{natbib}
\usepackage{tikz}
\usetikzlibrary{decorations.pathmorphing}
\usetikzlibrary{decorations.pathreplacing}
\usetikzlibrary{patterns}

\usepackage{lineno}

\newtheorem{theorem}{Theorem}
\newtheorem{proposition}{Proposition}
\newtheorem{corollary}{Corollary}
\newtheorem{lemma}{Lemma}

\theoremstyle{definition}

\theoremstyle{definition}
\newtheorem{definition}{Definition}
 
\theoremstyle{remark}
\newtheorem{remark}{Remark}

\title{Resource Leveling: Complexity of a UET two-processor scheduling variant and related problems}
\date{}
\author[1,2]{Bendotti, Pascale}
\author[2]{Brunod Indrigo, Luca}
\author[2]{Chrétienne, Philippe}
\author[2, 3]{Escoffier, Bruno}
\affil[1]{\small EDF R\&D, 7 boulevard Gaspard Monge, 91120 Palaiseau, France}
\affil[2]{\small Sorbonne Université, CNRS, LIP6 UMR 7606, 4 place Jussieu, 75005 Paris, France}
\affil[3]{\small Institut Universitaire de France}

\begin{document}

\maketitle

\abstract{
    This paper mainly focuses on a resource leveling variant of a two-processor scheduling problem.
    The latter problem is to schedule a set of dependent UET jobs on two identical processors with minimum makespan.
    It is known to be polynomial-time solvable.

    In the variant we consider, the resource constraint on processors is relaxed and the objective is no longer to minimize makespan.
    Instead, a deadline is imposed on the makespan and the objective is to minimize the total resource use exceeding a threshold resource level of two.
    This resource leveling criterion is known as the \emph{total overload cost}.
    Sophisticated matching arguments allow us to provide a polynomial algorithm computing the optimal solution as a function of the makespan deadline.
    It extends a solving method from the literature for the two-processor scheduling problem.

    Moreover, the complexity of related resource leveling problems sharing the same objective is studied.
    These results lead to polynomial or pseudo-polynomial algorithms or $NP$-hardness proofs, allowing for an interesting comparison with classical machine scheduling problems.
}


\maketitle

\section{Introduction}


Most project scheduling applications involve renewable resources such as machines or workers.
A natural assumption is to consider that the amount of such resources is limited, as in the widely studied Resource Constrained Project Scheduling Problem (RCPSP).
In many cases however, the resource capacity can be exceeded if needed, yet at a significant cost, by hiring additional workforce for instance.
The field of scheduling known as resource leveling aims at modelling such costs:
while resource capacities are not a hard constraint, the objective function is chosen to penalize resource overspending or other irregularities in resource use.

\paragraph{Related works}
Resource leveling is a well studied topic in recent literature, as a variant of the RCPSP -- see \cite{Hartmann2022} for a survey.
Among the various leveling objective functions that are proposed, a very natural one is the total overload cost, used in \cite{Rieck2012, Rieck2015, Bianco2016, Atan2018, Verbeeck2017}.
The total overload cost, that is the resource use exceeding a given level, is suitable for modelling the cost of mobilizing supplementary resource capacity, the level representing a base capacity that should ideally not be exceeded.
Considering a single resource of level $L$ and denoting $r_\tau$ the resource request at time step $\tau$, the function writes $\sum_{\tau} \max(0, r_\tau - L)$ (see the hatched part in Figure~\ref{fig:objective_function}).
Other notable examples of leveling objective functions introduced in the literature are weighted square capacities \citep{Christodoulou2015, Ponz2017a, Ponz2017b, Rieck2012, Rieck2015, Qiao2018, Li2018a, Li2018b}, squared changes in resource request \citep{Qiao2018, Ponz2017a}, absolute changes in resource request \citep{Bianco2017, Ponz2017a, Rieck2015}, resource availability cost \citep{Ponz2017a, Rodrigues2015, Van2013, Zhu2017} and squared deviation from a threshold \citep{Qiao2018}.
\begin{figure}[h!]
    \centering
    \scalebox{0.9}{\includegraphics{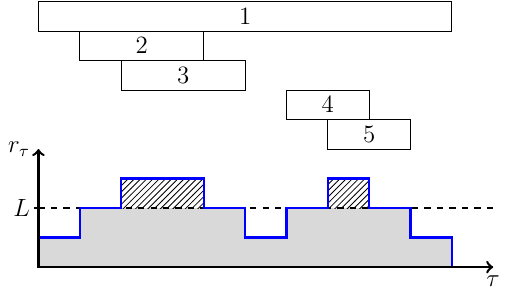}}
    \caption{Illustration of function $F$ for an instance with five jobs and resource level $L = 2$.}
    \label{fig:objective_function}
\end{figure}
\\
In terms of solving methods, the literature mainly provides heuristics \citep{Woodworth1975, Christodoulou2015,Zhu2017,Atan2018, Drotos2011}, metaheuristics \citep{Ponz2017a,Qiao2018,Li2018a,Van2013,Verbeeck2017} and exact methods  \citep{Easa1989,Demeulemeester1995, Ponz2017b, Rieck2012, Rieck2015,Bianco2017,Rodrigues2015, Bianco2016,Atan2018, Drotos2011}, but lacks theoretical results on the computational complexity of these problems. 
This work aims at providing such complexity results.

\paragraph{Contributions}
In \cite{Neumann1999}, resource leveling problems under generalized precedence constraints and with various objectives, including total overload cost, are shown to be strongly $NP$-hard using reductions from $NP$-hard machine scheduling problems of the literature.
The present work provides $NP$-hardness proofs based on the same idea, yet it tackles a large panel of more specific problems, including fixed-parameter special cases, in order to refine the boundary between $NP$-hardness and polynomial-time tractability. 
To some extent, resource leveling with total overload cost minimization shares similarities with late work minimization \citep{Sterna2011} in that portions of jobs exceeding a limit (in resource use or in time respectively) are penalized.
In that respect, a special case of resource leveling problem is proven equivalent to a two-machine late work minimization problem (see Section \ref{section:L2_Cmax_F}).
The latter late work minimization problem is studied in \cite{Chen2016} where it is shown to be $NP$-hard and solvable in pseudo-polynomial time.
Further approximation results on the early work maximization version of the same problem include a PTAS \citep{Sterna2017} and an FPTAS \citep{Chen2020}.
A similar equivalence between a resource leveling problem and a late work minimization problem is shown in \cite{Gyorgyi2020} by exchanging the roles of time and resource use (each time step becoming a machine and conversely).
A common approximation framework for both problems is proposed as well as an $NP$-hardness result that is included as such in this work to complement the range of investigated problems (see Section \ref{section:table}).

The core result in this work is a polynomial-time algorithm to minimize the total overload cost with resource level $L = 2$ in the case of unit processing times and precedence constraints.
This problem can be seen as the leveling counterpart of the well-known machine scheduling problem denoted $P2 | prec, p_i = 1 | C_{max}$ in the standard three-field Graham notation \citep{Graham1979}.
The algorithm relies on polynomial methods to solve $P2 | prec, p_i = 1 | C_{max}$ proposed in \cite{Fujii1969} and \cite{Coffman1971}.
Other resource leveling problems with various classical scheduling constraints are investigated: bounded makespan, precedence constraints (arbitrary or restricted to in-tree precedence graphs), release and due dates (with or without preemption).
Polynomial algorithms are provided for some cases, notably using flow reformulations or inspired from machine scheduling algorithms.

\paragraph{Outline}
This paper is organized as follows.
Section \ref{section:problem_description_and_notations} gives a general description of resource leveling problems and defines notations.
Section \ref{section:L2_prec} describes the polynomial-time algorithm for the main problem studied in this work: optimizing the total overload cost with level $L = 2$, unit processing times and precedence constraints.
Section \ref{section:other_polynomial_cases} provides complementary tractable cases among resource leveling problems with classical scheduling constraints.
Section \ref{section:NP_hardness_results} gives $NP$-hardness results.
Section \ref{section:table} summarizes the complexity results obtained in this work.
Section \ref{section:conclusion} gives elements of conclusion and perspectives.
The main notations used in this work are listed in Table \ref{tab:notations} at the end of the article.

\section{Problem description and notations} \label{section:problem_description_and_notations}

Even though this work deals with one main problem with specific constraints and parameters, this section gives a rather general description of resource leveling problems as well as convenient notations to designate them.
This will prove to be useful in differentiating the various other problems for which complexity results are also provided.

The scheduling problems considered in this work involve a single resource and a set of jobs $J$ with processing times $p \in \mathbb{N}^J$ and resource consumptions $c \in \mathbb{N}^J$.
A schedule is a vector of job starting times $x \in \mathbb{R}_+^J$.
The starting time of job $i \in J$ in schedule $x$ is denoted $x_i$.
Note that since parameters are integer, solutions with integer dates are dominant for the considered objective function in the non-preemptive case \citep{Neumann1999}.
In the sequel, only schedules with integer values will therefore be considered (with the exception of Section \ref{section:L2_ri_di_pmtn_F}).
For an integer date $\tau \in \mathbb{N}$, \emph{time step} $\tau$ will designate the time interval of size one starting at $\tau$, namely $[\tau, \tau+1)$.
\\
For the sake of readability, the usual three-field Graham notation $\alpha | \beta | \gamma$  introduced in \cite{Graham1979} for scheduling problems will now be extended with resource leveling parameters.

The $\alpha$ field is used to describe the machine environment of a scheduling problem.
In the context of resource leveling, the machine environment parameter is replaced with a resource level $L \in \mathbb{N}$.
In the sequel, the resource leveling problems with resource level $L$ are denoted $L | . | .$.
Following the case of machine scheduling, notations such as $L1 | . | .$ and $L2 | . | .$ will be used when $L$ is fixed.

The $\beta$ field contains information about constraints and instance specificities.
Since resource consumptions $c \in \mathbb{N}^J$ are introduced for resource leveling instances, the $\beta$ field should allow for restrictions on those values.
For example, $c_i = 1$ will denote the case of unit $c_i$ values.
In this work, the constraints of the $\beta$ field typically include a deadline on the makespan, denoted $C_{max} \leq M$.

The $\gamma$ field describes the objective function.
In this work, the amount of resource that fits under resource level $L$ (see the grey area in Figure~\ref{fig:objective_function}) is used as objective function, denoted $F$, instead of the total overload cost itself.
The two quantities being complementary, the complexity is equivalent and the minimization of the total overload is turned into a maximization problem.
This eases the interpretation of the criterion in terms of the structures used for solution methods (e.g., size of a matching or value of a flow).
Given an integer schedule $x$, denoting $r_\tau(x)$ the amount of resource required at time step $\tau$ in $x$ and $C_{max}(x)$ the makespan of $x$, function $F$ writes:
\[F(x) = \sum_{\tau = 0}^{C_{max}(x) - 1} \min(L, r_\tau(x))\]
Note that $F$ generalizes to non-integer schedules as follows:
\[F(x) = \int_{\tau = 0}^{C_{max}(x)} \min(L, r(\tau, x))d\tau\]
where $r(\tau, x)$ is the resource use at time $\tau$ in $x$.

\section{A polynomial algorithm for $L2 | prec, C_{max} \leq M, c_i = 1, p_i = 1 | F$} \label{section:L2_prec}

This section is dedicated to the core result of this work: solving Problem $L2 | prec, C_{max} \leq M, c_i = 1, p_i = 1 | F$ in polynomial time.
Given a set of jobs $J$ with unit processing times and unit resource consumption and a precedence graph $G = (J, \mathcal{A})$ with set of arcs $\mathcal{A}$, the problem is to find a feasible schedule $x$ with makespan at most $M$ such that $F(x)$ is maximized for a resource level $L = 2$.
Figure \ref{fig:L2_UET_precedence_graph_and_schedule} shows an example of precedence graph with optimal schedules for $M = 5$, $M = 6$ and $M = 7$.

\begin{figure}[h!]
    \centering
    \begin{subfigure}[c]{0.45\textwidth}
        \centering
        \scalebox{0.9}{\includegraphics{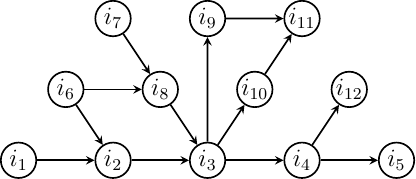}}
        \caption{Precedence graph $G$}
        \label{fig:L2_UET_precedence_graph}
    \end{subfigure}
    \begin{subfigure}[c]{0.45\textwidth}
        \centering
        \scalebox{0.9}{\includegraphics{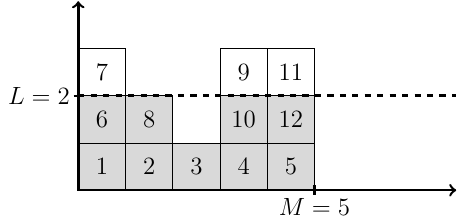}}
        \caption{Optimal schedule for $M = 5$}
        \label{fig:L2_UET_schedule_0}
    \end{subfigure}
    \begin{subfigure}[c]{0.45\textwidth}
        \centering
        \scalebox{0.9}{\includegraphics{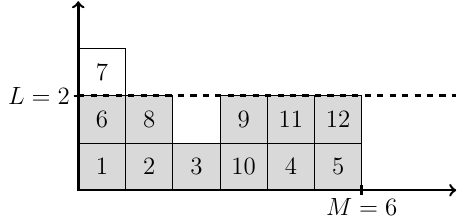}}
        \caption{Optimal schedule for $M = 6$}
        \label{fig:L2_UET_schedule_1}
    \end{subfigure}
    \begin{subfigure}[c]{0.45\textwidth}
        \centering
        \scalebox{0.9}{\includegraphics{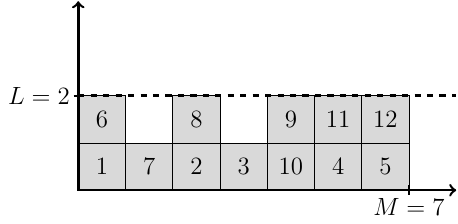}}
        \caption{Optimal schedule for $M = 7$}
        \label{fig:L2_UET_schedule_2}
    \end{subfigure}
    \caption{Example of precedence graph and optimal schedule for $L2 | prec, C_{max} \leq M, c_i = 1, p_i = 1 | F$}
    \label{fig:L2_UET_precedence_graph_and_schedule}
\end{figure}

For convenience, it will be assumed w.l.o.g. that there are no idling time steps, i.e., time steps where no job is processed, before the makespan of considered schedules.
With the above assumption, a schedule can be seen as a sequence of columns, a column being a set of jobs scheduled at the same time step.
Note that if two jobs $i$ and $j$ belong to the same column, there can be no path from $i$ to $j$ or from $j$ to $i$ in the precedence graph, in other words, $i$ and $j$ must be independent.

Problem $L2 | prec, C_{max} \leq M, c_i = 1, p_i = 1 | F$ is closely related to its classical counterpart in makespan minimization: $P2 | prec, p_i=1 | C_{max}$.
Among different approaches of the literature to solve $P2 | prec, p_i=1 | C_{max}$ in polynomial time, a method relying on maximum matchings between independent jobs is proposed in \cite{Fujii1969}.
Since independence between jobs is required for them to be scheduled at the same time step, matching structures naturally appear in a two-machine environment.
The results of this section will show that this is still true in the context of resource leveling.
Another approach is proposed in \cite{Coffman1971} with a better time complexity.
The idea of this second method is to build a priority list of the jobs based on the precedence graph and then apply the corresponding list-scheduling algorithm.
The optimality is proved using a block decomposition of the resulting schedule.
In broad outline, the convenient matching structure of the method from \cite{Fujii1969} is used to prove the existence of schedules that satisfy certain objective values while the more efficient method from \cite{Coffman1971} is used to actually build those schedules.

Two graphs will be considered as support for matchings, both of them being deduced from the precedence graph $G$, as illustrated in Figure \ref{fig:L2_UET_graphs}.
First, let $\widetilde{G} = (J, E)$ be the \emph{independence graph} where $E = \{\{i, j\} \in J^2 \mid i \text{ independent from } j\}$ -- see Figure \ref{fig:L2_UET_independence_graph}.
A critical path in a longest path in the precedence graph.
It is a well known result that the sum of processing times along a critical path is the minimum feasible makespan with respect to precedence constraints.
Let then $P$ be a critical path in $G$.
Consider the bipartite graph $\widetilde{G}_P = (P \cup (J \setminus P), E_P)$ where $E_P = \{\{i, j\} \in E| i \in P, j \in J \setminus P\}$, as shown in Figure \ref{fig:L2_UET_bipartite_independence_graph} where the jobs of $P$ are in gray.
Note that in Figure \ref{fig:L2_UET_graphs}, independence relations are represented with dashed lines as opposed to solid lines for precedence relations.
\begin{figure}[h!]
    \centering
    \begin{subfigure}[c]{0.45\textwidth}
        \centering
        \scalebox{0.9}{\includegraphics{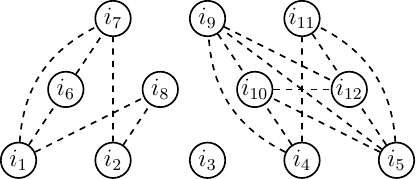}}
        \caption{Independence graph $\widetilde{G}$}
        \label{fig:L2_UET_independence_graph}
    \end{subfigure}
    \begin{subfigure}[c]{0.45\textwidth}
        \centering
        \scalebox{0.9}{\includegraphics{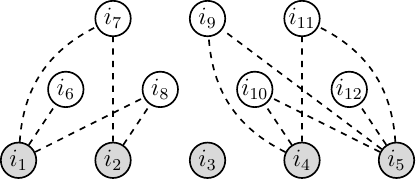}}
        \caption{Bipartite independence graph $\widetilde{G}_P$}
        \label{fig:L2_UET_bipartite_independence_graph}
    \end{subfigure}
    \caption{Examples of graphs for $L2 | prec, C_{max} \leq M, c_i = 1, p_i = 1 | F$}
    \label{fig:L2_UET_graphs}
\end{figure}

Let $m^*_P$ and $m^*$ denote the sizes of a maximum matching in $\widetilde{G}_P$ and $\widetilde{G}$ respectively.
The inequality $m^*_P \leq m^*$ is always verified since $E_P$ is a subset of $E$.
In the example of Figure \ref{fig:L2_UET_graphs}, $m^*_P = 4$ and $m^* = 5$.
\\
The aim is to show that the optimal value of objective function $F$ can be computed for any $M$ using only $m^*_P$, $m^*$, $|J|$ and $|P|$ -- and that an associated optimal schedule can be computed in polynomial time.
More precisely, it will be shown that the optimal objective value as a function of $M$, denoted $F^*$, is piecewise linear with three distinct segments as shown in Figure \ref{fig:L2_prec_objective_function_of_M}.
Note that if $M < |P|$ there is no feasible schedule.

\begin{figure}[h!]
    \centering
    \scalebox{0.7}{\includegraphics{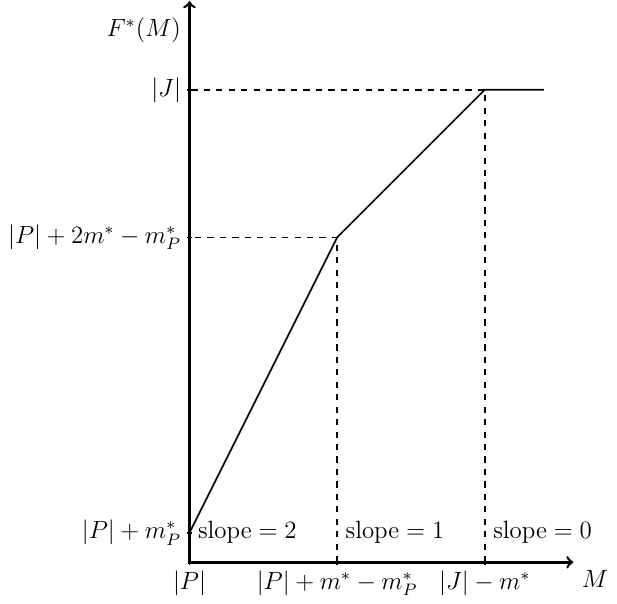}}
    \caption{Optimal objective value as a function of $M$}
    \label{fig:L2_prec_objective_function_of_M}
\end{figure}

The constant part of the function on the right coincides with the result of \cite{Fujii1969} stating that the minimum makespan to schedule the project on two machines is $|J| - m^*$.
For $M \geq |J| - m^*$, $F^*$ equals its upper bound $|J|$ -- all jobs fit under resource level $L = 2$.
The contribution of this work is therefore to provide the values of $F^*(M)$ for $|P|\leq M \leq |J| - m^*$.

Note that the schedule of Figures \ref{fig:L2_UET_schedule_0} corresponds to an optimal solution for $M = |P| = 5$.
Its value is $9$, which is indeed equal to $|P| + m^*_P = 5 + 4$.
If the deadline is increased by one, i.e., $M = 6$, then two more jobs can be added under the resource level two, i.e., the optimal value is $11$, as shown in Figure \ref{fig:L2_UET_schedule_1}.
It corresponds to the first breaking point in Figure \ref{fig:L2_prec_objective_function_of_M}, as $M = 6 = |P| + m^*-m^*_P$ and $11=|P|+2m^*-m^*_P$.
As it can be seen in Figure \ref{fig:L2_UET_schedule_2}, when $M = 7$, all jobs can be scheduled under the resource level.
It corresponds to the second breaking point in Figure \ref{fig:L2_prec_objective_function_of_M}, as $M = 7 = |J| - m^*$.

The proposed approach is in two steps:
\begin{enumerate}
    \item Find an optimal solution for $M = |P|$;
    \item Find recursively an optimal solution for $M + 1$ from an optimal solution for $M$.
\end{enumerate}

The following definition introduces the \emph{augmenting sequence of a schedule} in the independence graphs $\widetilde{G}$ and in the bipartite independence graph $\widetilde{G}_P$.

\begin{definition}[augmenting sequence] \label{def:augmenting_sequence}
    A sequence $(i_0, j_1, i_1, j_2, i_2, \dots, j_r, i_r, j_{r+1})$ of distinct jobs is an \emph{augmenting sequence} for $(\widetilde{G}, x)$ (resp. for $(\widetilde{G}_P, x)$) if:
    \begin{itemize}
        \item[-] $(i_0, j_1, i_1, j_2, i_2, \dots, j_r, i_r, j_{r+1})$ is a path in $\widetilde{G}$ (resp. in $\widetilde{G}_P$);
        \item[-] the number of jobs scheduled at $x_{i_0}$ is not two;
        \item[-] the number of jobs scheduled at $x_{j_{r+1}}$ is not two;
        \item[-] for every $q \in \{1, \dots, r\}$, $x_{i_q} = x_{j_q}$.
    \end{itemize}
\end{definition}

The name augmenting sequence deliberately recalls the augmenting path used in the context of matching problems.
It will be seen in the sequel that both notions actually coincide to some extent since an augmenting sequence defines an augmenting path for a certain matching of jobs.

An example of augmenting sequence is illustrated in Figure \ref{fig:augmenting_sequence}.
The jobs of the sequence are surrounded by circles and connected with dashed arrows.

\begin{figure*}[h!]
    \centering
    \includegraphics{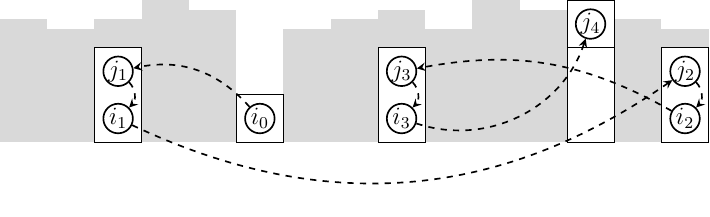}
    \caption{Example of augmenting sequence}
    \label{fig:augmenting_sequence}
\end{figure*}

\subsection{Solution for $M = |P|$} \label{section:resolution_Mmin}

The first step is to solve the problem for the minimum relevant value for $M$, that is to say when $M$ equals the length of a critical path in the precedence graph.
\\
The following definition describes an elementary operation, illustrated in Figure \ref{fig:elementary_operation}, that will be used in the sequel to improve the value of objective function $F$ incrementally.
\begin{definition}[elementary operation] \label{def:elementary_operation}
    Let $x$ be a feasible schedule for instance $I = (G, M)$ with $M = |P|$, where $P$ is a critical path of $G$.
    Let $a, b \in J$ be two independent jobs such that $x_a < x_b$ and such that $a$ is the only job scheduled at $x_a$.
    \\
    \underline{Decomposition:} Let us denote $\tau_1 < \tau_2 < \dots < \tau_K$ the time steps in $\{x_a + 1, \dots, x_b - 1\}$ where at least one predecessor of $b$ is scheduled, as well as $\tau_0 = x_a$ and $\tau_{K+1} = x_b$ for convenience.
    For each $k \in \{1, \dots, K\}$, the jobs scheduled at $\tau_k$ are partitioned into $\alpha_k$ and $\beta_k$, where $\beta_k$ are the predecessors of $b$ and $\alpha_k$ are the remaining jobs.
    Let us also denote, for each $k \in \{1, \dots, K+1\}$,  $A_k = \{i \in J \mid \tau_{k-1} < x_i < \tau_k\}$, the sets of jobs scheduled in between the $\tau_k$. 
    \\
    \underline{Translation:} The jobs of $\beta_k$ are translated from $\tau_k$ to $\tau_{k-1}$ for each $k \in \{1, \dots, K\}$ and job $b$ from $\tau_{K+1}$ to $\tau_K$.
    \begin{figure*}[h!]
        \centering
        \begin{subfigure}[c]{\textwidth}
            \centering
            \scalebox{0.9}{\includegraphics{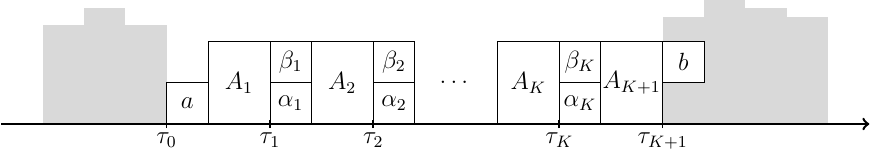}}
            \caption{Decomposition}
            \label{fig:elementary_operation_decomposition}
        \end{subfigure}
        \begin{subfigure}[c]{\textwidth}
            \centering
            \scalebox{0.9}{\includegraphics{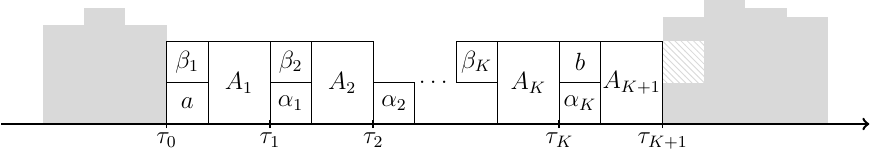}}
            \caption{Translation}
            \label{fig:elementary_operation_translation}
        \end{subfigure}
        \caption{Elementary operation}
        \label{fig:elementary_operation}
    \end{figure*}
\end{definition}
The decomposition and translation steps of Definition \ref{def:elementary_operation} are illustrated in Figure \ref{fig:elementary_operation_decomposition} and Figure \ref{fig:elementary_operation_translation} respectively.

The elementary operation of Definition \ref{def:elementary_operation} will now be shown to yield a feasible schedule and its impact on objective function $F$ will be quantified.
\begin{proposition} \label{prop:elementary operation}
    Let $x$ be a feasible schedule for instance $I = (G, M)$ with $M = |P|$, where $P$ is a critical path of $G$.
    Let $a, b \in J$ be two independent jobs such that $x_a < x_b$ and such that $a$ is the only job scheduled at $x_a$.
    \\
    The schedule $x'$ resulting from the elementary operation on $a$ and $b$ is feasible and it verifies:
    \[F(x')=
    \begin{cases}
        F(x) &\text{if exactly two jobs are}\\
             &\text{scheduled at } x_b \text{ in } x\\
        F(x) + 1 &\text{otherwise} 
    \end{cases}
    \]
\end{proposition}

\begin{proof}
    Let us first show that schedule $x'$ is feasible.
    It is clear that the elementary operation does not increase the makespan of the schedule, thus respecting the makespan bound $M = |P|$.
    By construction, all predecessors of job $b$ scheduled between $x_a$ and $x_b$ are in $\bigcup_{k = 1}^{K} \beta_k$.
    As a consequence, there is no precedence arc going from a job of $\{a\} \cup \bigcup_{k = 1}^K \alpha_k \cup \bigcup_{k = 1}^{K + 1} A_k$ to a job of $\{b\} \cup \bigcup_{k = 1}^{K} \beta_k$.
    The translation of the elementary operation does no violate any precedence constraint.
    \\
    Let us now quantify the impact of the elementary operation on the objective function.
    Since the makespan bound is $M = |P|$, there is exactly one job of critical path $P$ at each time step in $\{0, 1, \dots, M-1\}$.
    In particular, since $a$ is the only job scheduled at $x_a$, it necessarily belongs to $P$.
    Furthermore, for every $k \in \{1, \dots, K\}$, a job of $P$ is scheduled at $\tau_k$.
    This job of $P$ is a successor of $a$ and therefore cannot be in $\beta_k$, otherwise $a$ and $b$ would not be independent, by transitivity.
    The job of $P$ must then be in $\alpha_k$, so $|\alpha_k| \geq 1$.
    Also, $|\beta_k| \geq 1$ for every $k \in \{1, \dots, K\}$ by construction.
    As a consequence:
    \begin{itemize}
        \item[-] For every $k \in \{1, \dots, K\}$, at least two jobs are scheduled at $\tau_k$ in $x'$, which was already the case in $x$
        \item[-] At least two jobs are scheduled at $x_a$ in $x'$ while there was only one in $x$
        \item[-] One less job is scheduled at $x_b$ in $x'$
        \item[-] Other time steps remain unchanged  
    \end{itemize}
    Job $b$ cannot be the only job scheduled at $x_b$, otherwise $b$ would belong to critical path $P$ and would not be independent from job $a$.
    The difference between $F(x)$ and $F(x')$ comes down to whether there are two or more jobs scheduled at $x_b$ in $x$:
    \begin{itemize}
        \item[-] If exactly two jobs are scheduled at $x_b$ in $x$, then the contribution to the objective function is increased by one in $x_a$ but decreased by one in $x_b$, so $F(x') = F(x)$
        \item[-] If at least three jobs are scheduled at $x_b$ in $x$, then the contribution to the objective function is increased by one in $x_a$ and remains unchanged in $x_b$, so $F(x') = F(x) + 1$
    \end{itemize}
\end{proof}

\begin{remark}
    Definition \ref{def:elementary_operation} and Proposition \ref{prop:elementary operation} assume that jobs $a$ and $b$ in the elementary operation are such that $x_a < x_b$.
    The case where $x_b < x_a$ can actually be handled exactly the same way, up to an inversion of precedence relations.
\end{remark}

The possibility to improve the objective value of a schedule with makespan $|P|$ using elementary operations will now be linked to the existence of a particular augmenting sequence.

\begin{lemma} \label{lemma:augmenting_sequence}
    Let $x$ be a feasible schedule for instance $I = (G, M)$ with $M = |P|$, where $P$ is a critical path of $G$.
    If there exists an augmenting sequence $\rho = (i_0, j_1, i_1, j_2, i_2, \dots, j_r, i_r, j_{r+1})$ for $(\widetilde{G}_P, x)$ such that:
    \begin{itemize}
        \item[-] $i_0$ is the only job scheduled at $x_{i_0}$;
        \item[-] at least three jobs are scheduled at $x_{j_{r+1}}$;
    \end{itemize}
    then there exists a schedule $x'$ with makespan $|P|$ such that $F(x') = F(x)+1$.
    \\
    Furthermore, such a schedule can be computed from $x$ in polynomial time.
\end{lemma}

\begin{proof}
    The result will be proven using an induction on $r$.
    
    (initialization) For $r = 0$, the augmenting sequence reduces to two jobs $i_0$ and $j_1$ that are independent and scheduled at different dates.
    The elementary operation can be applied with $i_0$ as $a$ and $j_1$ as $b$.
    Since there are at least three jobs scheduled at $x_{j_1}$, the resulting schedule $x'$ is such that $F(x') = F(x) + 1$ according to Proposition \ref{prop:elementary operation}.

    (induction) Suppose that the property is verified up to the value $r-1$ for some $r \geq 1$.
    If for some $q \in \{1, \dots, r\}$ there are at least three jobs scheduled at $x_{i_q}$, then the sequence $(i_0, j_1, i_1, j_2, i_2, \dots, j_{q-1}, i_{q-1}, j_{q})$ is an augmenting sequence with the required properties on which the induction hypothesis can be applied.
    \\
    One can now suppose that for every $q \in \{1, \dots, r\}$, the only jobs scheduled at $x_{i_q}$ are $i_q$ and $j_q$.
    The idea is now to use the elementary operation of Definition \ref{def:elementary_operation} to reduce the size of the augmenting sequence.
    Assume that the elementary operation is applied using $i_0$ as $a$ and $j_1$ as $b$ and let $x''$ be the resulting schedule.
    If $(i_1, j_2, i_2, j_3, i_3, \dots, j_r, i_r, j_{r+1})$ is an augmenting sequence for $(\widetilde{G}_P, x'')$, then the induction can be applied.
    \\
    If not, there exists $q \geq 2$ such that $\{i_q, j_q\}$ coincide with $\alpha_k \cup \beta_k$ in the decomposition of Definition \ref{def:elementary_operation}, implying that $x''_{i_q} \neq x''_{j_q}$.
    In that case, schedule $x''$ would be irrelevant for the induction.
    Such a crossing configuration is illustrated in Figure \ref{fig:augmenting_path_crossing}.
    Moreover, let us show that $\alpha_k = \{i_q\}$ and $\beta_k = \{j_q\}$.
    Since the augmenting sequence $\rho$ is a path in the bipartite graph $\widetilde{G}_P$, it alternates between jobs of the critical path $P$ and jobs of $J \setminus P$.
    In a schedule of makespan $|P|$, exactly one job belongs to $P$ at each time step, so $i_0$ that is the only job scheduled at $x_{i_0}$ belongs to $P$ and so do $i_1, i_2, \dots, i_r$.
    By construction, in the decomposition of Definition \ref{def:elementary_operation}, only block $\alpha_k$ can contain successors of $i_0 = a$.
    Therefore, $i_q$ that is a successor of $i_0$ since it is in the critical path must be in $\alpha_k$.
    The only other job scheduled at $x_{i_q}$ is $j_q$ and it must belong to $\beta_k$ that cannot be empty.
    \\
    Having $\beta_k = \{j_q\}$ guarantees that $j_q$ is independent from $a = i_0$, so $\{i_0, j_q\}$ is an edge in $\widetilde{G}_P$.
    A shortcut can then be operated from $i_0$ to $j_q$, as illustrated in Figure \ref{fig:augmenting_path_crossing_shortcut}, yielding sequence $(i_0, j_q, i_q, j_{q+1}, i_{q+1}, \dots, j_{r}, i_{r}, j_{r+1})$ that is an augmenting sequence on which the induction hypothesis applies.

    \begin{figure*}[h!]
        \centering
        \begin{subfigure}[c]{\textwidth}
            \centering
            \scalebox{0.9}{\includegraphics{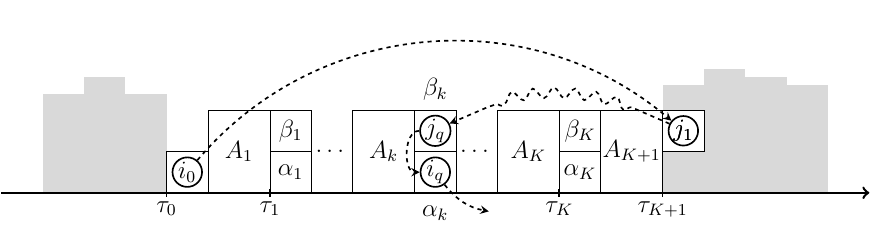}}
            \caption{Crossing configuration with $\{i_q\} = \alpha_k$ and $\{j_q\} = \beta_k$}
            \label{fig:augmenting_path_crossing}
        \end{subfigure}
        \begin{subfigure}[c]{\textwidth}
            \centering
            \scalebox{0.9}{\includegraphics{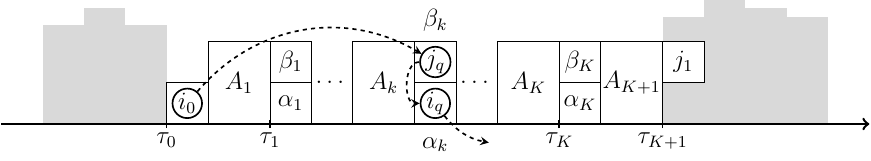}}
            \caption{Shortcut}
            \label{fig:augmenting_path_crossing_shortcut}
        \end{subfigure}
        \caption{Solving a crossing configuration in the augmenting sequence}
        \label{fig:solving_a_crossing_configuration_in_augmenting_path}
    \end{figure*}
\end{proof}

\begin{proposition} \label{prop:maximum_matching_Mmin}
    Any optimal solution $x^*$ of the instance $I = (G, M)$ with $M = |P|$ verifies $F(x^*) = |P| + m^*_P$.
    Furthermore, such an optimal schedule can be computed in polynomial time.
\end{proposition}

\begin{proof}
    Let us first show that the objective function value for a feasible schedule is at most $|P| + m^*_P$.
    Let $x$ be a feasible schedule with makespan $|P|$.
    Let $m$ be the number of time steps in $\{0, \dots, |P|-1\}$ with at least two jobs scheduled in $x$.
    It is clear that $F(x) = |P|+m$.
    By definition of a critical path, exactly one job of $P$ is scheduled at every time step of $\{0, \dots, |P|-1\}$.
    Given that two jobs scheduled at the same time step are necessarily independent, matching each job of the critical path with a job that is scheduled at the same date, when there is one, yields a matching in $\widetilde{G}_P$.
    In other words, a matching of size $m$ in $\widetilde{G}_P$ can be deduced from $x$, so $m \leq m^*_P$ and $F(x) \leq |P| + m^*_P$.
    \\
    It will now be shown that there exists a feasible schedule $x$ such that $F(x) = |P| + m^*_P$.
    Suppose that a given schedule $x$ is such that there are $m$ time steps in $\{0, \dots, |P|-1\}$ where at least two jobs are scheduled, with $m < m^*_P$.
    Again, $F(x) = |P|+m$ and it is possible to deduce from $x$ a matching $\mu$ of size $m$ in $\widetilde{G}_P$.
    This matching being not of maximum size, Berge's theorem \citep{Berge1957} states that there exists an augmenting path.
    Let $(i_0, j_1, i_1, j_2, i_2, \dots, j_r, i_r, j_{r+1})$ denote the sequence of jobs in the augmenting path, with $i_0 \in P$, $j_{r+1} \in J \setminus P$ and both being not matched in $\mu$.
    The following properties hold:
    \begin{itemize}
        \item[-] $i_0$ is a job of $P$ that is not matched in $\mu$, so it is the only one scheduled at $x_{i_0}$.
        \item[-] $j_{r+1}$ is a job of $J \setminus P$ that is not matched in $\mu$. 
        Since there is a job of $P$ at each time step, $j_{r+1}$ is scheduled at $x_{j_{r+1}}$ together with a job of $P$.
        This job of $P$ could be matched with $j_{r+1}$ in $\mu$ but it is not so, by construction of $\mu$, it must be matched with a third job.
        Thus, three jobs are scheduled at $x_{j_{r+1}}$.
        \item[-] $i_q$ and $j_q$ are matched together in $\mu$ so $x_{i_q} = x_{j_q}$ for every $q \in \{1, \dots, r\}$. 
    \end{itemize}
    Those properties imply that $(i_0, j_1, i_1, j_2, i_2, \dots, j_r, i_r, j_{r+1})$ is an augmenting sequence satisfying the requirements of Lemma \ref{lemma:augmenting_sequence}.
    Applying Lemma \ref{lemma:augmenting_sequence} gives that there exists a schedule $x'$ with $F(x') = F(x)+1$.
    As a conclusion, any feasible schedule of $I = (G, |P|)$ can be modified incrementally to reach an objective function value of $|P|+m^*_P$.
\end{proof}

\subsection{Solution for $M > |P|$}

The problem will now be solved for other values of $M$ based on the solution for $M = |P|$.
The idea is to apply transformations that increment the makespan of the initial schedule while keeping its optimality.
First, it will be shown that optimal solutions always reach makespan deadline $M$ when $M$ is in a certain interval.

\begin{lemma} \label{lemma:makespan_of_optimal_schedule}
    For any makespan deadline $M \in \{|P|, \dots, |J| - m^*\}$ any optimal schedule for $I = (G, M)$ has makespan exactly $M$.
\end{lemma}

\begin{proof}
    Let $M \in \{|P|, \dots, |J| - m^*\}$ and suppose that there exists an optimal schedule $x$ for $I$ of makespan $C_{max}(x) < M$.
    According to the result of \cite{Fujii1969}, the minimum makespan to schedule the project on two machines is $|J| - m^*$.
    This implies that there exists a time step $\tau$ where at least three jobs are scheduled in $x$.
    \\
    Consider the transformation illustrated in Figure \ref{fig:simple_elongation_single_job}.
    One of the jobs scheduled at $\tau$ can be scheduled at $\tau+1$ instead and the subsequent jobs can be scheduled one time step later.
    \\
    The resulting schedule $x'$ has makespan $C_{max}(x') = C_{max}(x) + 1 \leq M$ and objective value $F(x') = F(x) + 1$.
    This contradicts the optimality of $x$ since $x'$ is feasible for $I$ with a strictly higher objective value.
\end{proof}

\begin{figure*}[h!]
    \centering
    \scalebox{0.9}{\includegraphics{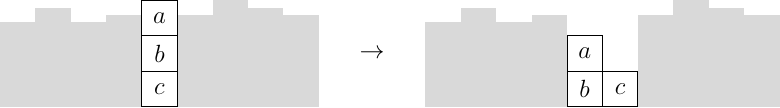}}
    \caption{Example of one-job elongation}
    \label{fig:simple_elongation_single_job}
\end{figure*}

\begin{lemma} \label{lemma:generalized_augmenting_path}
    Let $M \in \{|P|, \dots, |J| - m^* - 1\}$ and let $x$ be an optimal schedule for $I = (G, M)$. 
    Suppose that $x$ is such that, at each time step $\tau \in \{0, \dots, M-1\}$, either $1$, $2$ or $3$ jobs are scheduled.
    If there exists an augmenting sequence $\rho = (i_0, j_1, i_1, j_2, i_2, \dots, j_r, i_r, j_{r+1})$ for $(\widetilde{G}, x)$ satisfying $x_{i_0} \neq x_{j_{r+1}}$,
    then there exists a schedule $x'$ with makespan $M + 1$ such that $F(x') = F(x) + 2$.
    \\
    Furthermore, such a schedule can be computed from $x$ in polynomial time.
\end{lemma}

\begin{proof}
    Some useful properties of the augmenting sequence are first derived from basic transformations.
    \\
    An augmenting sequence can indeed be obtained in which $i_q$ and $j_q$ are the only jobs scheduled at $x_{i_q}$ for every $q \in \{1, \dots, r\}$.
    Suppose that three jobs are scheduled at $x_{i_q}$ for some $q \in \{1, \dots, r\}$, two cases are possible:
    \begin{itemize}
        \item[-] if $i_0$ is scheduled at $x_{i_q}$, then the sequence $(i_q, j_{q+1}, i_{q+1}, j_{q+2}, i_{q+2}, \dots, j_r, i_r, j_{r+1})$ is an augmenting sequence and satisfies $x_{i_q} \neq x_{j_{r+1}}$;
        \item[-] otherwise the sequence $(i_0, j_1, i_1, j_2, i_2, \dots, j_{q-1}, i_{q-1}, j_{q})$ is an augmenting sequence satisfying $x_{i_0} \neq x_{j_q}$.
    \end{itemize}
    In both cases, an augmenting sequence with the required properties can be obtained by truncating the initial one.
    Such truncations can be operated on the augmenting sequence until all intermediary time steps have exactly two jobs scheduled.

    It can also be ensured that no arc of the augmenting sequence jumps over a time step different from $x_{i_0}$ and $x_{j_{r+1}}$ with exactly three jobs scheduled.
    Indeed, suppose that an arc $(a, b)$ of the augmenting sequence is such that there exists a job $c$ with exactly three jobs scheduled at $x_c$ and w.l.o.g. $x_a < x_c < x_b$, as shown in Figure \ref{fig:generalized_augmenting_path_jump_cut}.
    Since $a$ and $b$ are independent, $c$ is either independent from $a$ or from $b$ -- or both.
    If $c$ is independent from $a$, then the augmenting sequence can be truncated after $a$ and ended with arc $(a, c)$.
    If $c$ is independent from $b$, then the augmenting sequence can be truncated before $b$ and started with arc $(c, b)$.
    \begin{figure*}[h!]
        \centering
        \scalebox{0.9}{\includegraphics{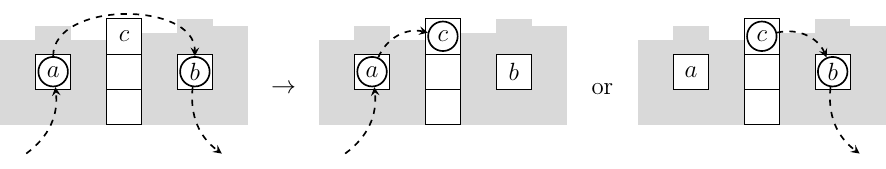}}
        \caption{Shortcut preventing a jump}
        \label{fig:generalized_augmenting_path_jump_cut}
    \end{figure*}
    \\
    One can now suppose that the augmenting sequence $\rho$ is such that:
    \begin{itemize}
        \item[-] $i_q$ and $j_q$ are the only jobs scheduled at $x_{i_q}$ for every $q \in \{1, \dots, r\}$;
        \item[-] for every arc $(i_q, j_{q+1})$ and every time step $\tau \in \{x_{i_q}, \dots, x_{j_{q+1}}\} \setminus \{x_{i_0}, x_{j_{r+1}}\}$, at most two jobs are scheduled at $\tau$.
    \end{itemize}
    Let $\tau_{min} = \min_{i \in \rho} x_i$ and $\tau_{max} = \max_{i \in \rho} x_i$.
    For every time step $\tau \in \{\tau_{min}, \dots, \tau_{max}\} \setminus \{x_{i_0}, x_{j_{r+1}}\}$, either $1$ or $2$ jobs are scheduled at $\tau$.
    Consider the subset of jobs $J' = \{i \in J \mid \tau_{min} \leq x_i \leq \tau_{max}\} \setminus \{i_0, j_{r+1}\}$.
    \\
    Let $\mu$ denote the matching on $J'$ obtained by pairing the jobs scheduled at the same time step.
    It is clear that $\mu \cup \{\{i_q, j_{q+1}\} \mid q \in \{0, \dots, r\}\} \setminus \{\{j_q, i_q\} \mid q \in \{1, \dots, r\}\}$ is a matching between independent jobs of $J' \cup \{i_0, j_{r+1}\}$ of size $|\mu| + 1$.
    Using the result of \cite{Fujii1969}, it is possible to schedule the jobs of $J' \cup \{i_0, j_{r+1}\}$ on two machines using at most $|J'| + 2 - (|\mu| + 1) = |J'| - |\mu| + 1$ time steps.
    Three cases can be distinguished, only one of which is actually possible:
    \begin{enumerate}
        \item If $i_0$ and $j_{r+1}$ are both the only jobs scheduled on their time steps in $x$, the jobs of $J' \cup \{i_0, j_{r+1}\}$ are scheduled on $|J'| - |\mu| + 2$ time steps in $x$.
        This contradicts the optimality of $x$, since, using \cite{Fujii1969}, there exists a schedule of makespan $M-1$ with the same objective value as $x$, which is impossible according to Lemma \ref{lemma:makespan_of_optimal_schedule}.
        \item If either $i_0$ or $j_{r+1}$ is scheduled alone on its time step -- the other being on a time step with exactly three jobs, the jobs of $J' \cup \{i_0, j_{r+1}\}$ are scheduled on $|J'| - |\mu| + 1$ time steps in $x$.
        This contradicts the optimality of $x$ since, using \cite{Fujii1969}, all jobs of $J' \cup \{i_0, j_{r+1}\}$ can be scheduled on the same number of time steps without exceeding the resource level.
        \item If both $i_0$ and $j_{r+1}$ are scheduled on time steps with exactly three jobs, the jobs of $J' \cup \{i_0, j_{r+1}\}$ are scheduled on $|J'| - |\mu|$ time steps in $x$.
        Using \cite{Fujii1969}, these jobs can be scheduled on $|J'| - |\mu| + 1$ time steps without exceeding the resource level, so the objective function can be increased by $2$.
        This is the expected result.
    \end{enumerate}
\end{proof}

\begin{lemma} \label{lemma:upper_bound_optimal_objective_function_1}
    For any $M \in \mathbb{N}$, $M \geq |P|$, the optimal objective function verifies:
    $$F^*(M) \leq 2(M - |P|) + |P| + m^*_P$$
\end{lemma}

\begin{proof}
    Using Proposition \ref{prop:maximum_matching_Mmin} and by definition of $m^*_P$, in any feasible schedule, at most $m^*_P$ jobs of $P$ are scheduled on the same time step as a job of $J \setminus P$.
    Equivalently, at least $|P| - m^*_P$ jobs of $P$ are alone on their time step in any feasible schedule.
    \\
    This implies that $F^*(M) \leq 2M - (|P| - m^*_P)$ for any $M \in \mathbb{N}$, $M \geq |P|$, which can be rewritten as $F^*(M) \leq 2(M - |P|) + |P| + m^*_P$.
\end{proof}

\begin{lemma} \label{lemma:upper_bound_optimal_objective_function_2}
    For any $M \in \mathbb{N}$, $M \geq |P|$, the optimal objective function verifies:
    $$F^*(M) \leq M + m^*$$
\end{lemma}

\begin{proof}
    Let $M \in \mathbb{N}$, $M \geq |P|$, and suppose that there exists a feasible schedule $x$ for $I = (G, M)$ such that $F(x) = M + m$.
    At least $m$ time steps have two or more jobs scheduled in $x$.
    A matching between independent jobs of size at least $m$ can then be deduced from $x$.
    By definition of $m^*$ as the maximum size for a matching between independent jobs, $m \leq m^*$ and $F(x) \leq m^* + M$.
    Inequality $F(x) \leq m^* + M$ is then verified for any feasible schedule so $F^*(M) \leq M + m^*$.
\end{proof}

The following theorem allows for the computation of $F^*$ as a function of $M$ as illustrated in Figure \ref{fig:L2_prec_objective_function_of_M}.

\begin{theorem} \label{theorem:L2_UET_prec}
    The optimal objective function value $F^*$ as a function of makespan deadline $M$ is piecewise linear and defined as follows:
    \begin{itemize}
        \item[-] For $M \in \{|P|, \dots, |P| + m^* - m^*_P\}$: $F^*(M) = 2(M - |P|) + |P| + m^*_P$;
        \item[-] For $M \in \{|P| + m^* - m^*_P, \dots, |J| - m^*\}$: $F^*(M) = M + m^*$;
        \item[-] For $M \in \mathbb{N}$, $M \geq |J| - m^*$: $F^*(M) = |J|$.
    \end{itemize}
    Furthermore, an optimal schedule can be computed for any $M \geq |P|$ in polynomial time.
\end{theorem}

\begin{proof}
    Let us prove by induction on $M$ that there exists a schedule $x^M$ of makespan $M$ such that $F(x^M) = 2(M - |P|) + |P| + m^*_P$ for every $M \in \{|P|, \dots, |P| + m^* - m^*_P\}$.
    \\
    The case of $M = |P|$ was solved in Section \ref{section:resolution_Mmin}.
    \\
    Suppose that the property holds for some $M \in \{|P|, \dots, |P| + m^* - m^*_P - 1\}$ and let $x^M$ be the associated schedule.
    If there exists a time step with at least $4$ jobs scheduled in $x^M$, two of those jobs can be scheduled on the next time step while all jobs scheduled after are postponed by one time unit, as illustrated in Figure \ref{fig:simple_elongation}.
    This yields an optimal schedule $x^{M+1}$.
    \begin{figure*}[h!]
        \centering
        \scalebox{0.9}{\includegraphics{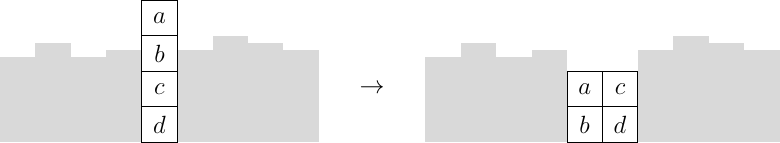}}
        \caption{Example of two-job elongation}
        \label{fig:simple_elongation}
    \end{figure*}
    If not, $x^M$ satisfies the requirement of Lemma \ref{lemma:generalized_augmenting_path} since either $1$, $2$ or $3$ jobs are scheduled at each time step $\tau \in \{0, \dots, M-1\}$.
    Let $\mu$ denote a matching obtained by selecting a pair of jobs for each time step where at least two jobs are scheduled in $x^M$.
    Since $F(x^M) = 2(M - |P|) + |P| + m^*_P$ according to the induction hypothesis, there are $M - |P| + m^*_P$ such time steps.
    It is then clear that $\mu$ is a matching between independent jobs of size $M - |P| + m^*_P$.
    This matching is not maximum because $M - |P| + m^*_P \leq |P| + m^* - m^*_P - 1 - |P| + m^*_P < m^*$.
    By Berge's Theorem, there exists an augmenting path $\rho = (i_0, j_1, i_1, j_2, i_2, \dots, j_r, i_r, j_{r+1})$ in $\widetilde{G}$ such that $\{i_q, j_q\} \in \mu$ for each $q \in \{1, \dots, r\}$ and such that neither $i_0$ nor $j_{r+1}$ are matched in $\mu$.
    Notice the following:
    \begin{itemize}
        \item[-] If exactly two jobs were scheduled at $x^M_{i_0}$, $i_0$ would be in $\mu$. Either $1$ or $3$ jobs are scheduled at $x^M_{i_0}$.
        \item[-] If exactly two jobs were scheduled at $x^M_{j_{r+1}}$, $j_{r+1}$ would be in $\mu$. Either $1$ or $3$ jobs are scheduled at $x^M_{j_{r+1}}$.
        \item[-] By construction of $\mu$, $x^M_{i_q} = x^M_{j_q}$ for every $q \in \{1, \dots, r\}$.
    \end{itemize}
    Furthermore, if $x_{i_0}$ and $x_{j_{r+1}}$ were equal, three jobs would be scheduled at $x_{i_0}$ and at least one job in $\{i_0, j_{r+1}\}$ would be matched in $\mu$.
    As a consequence, $x_{i_0} \neq x_{j_{r+1}}$.
    All the requirements of Lemma \ref{lemma:generalized_augmenting_path} are now satisfied by $x^M$ and $\rho$.
    Applying Lemma \ref{lemma:generalized_augmenting_path} gives that there exists a schedule $x^{M+1}$ of makespan $M+1$ such that $F(x^{M+1}) = 2(M + 1 - |P|) + |P| + m^*_P$.
    This value reaches the upper bound of Lemma \ref{lemma:upper_bound_optimal_objective_function_1} so $x^{M+1}$ is optimal.

    Let us prove by induction on $M$ that there exists a schedule $x^M$ of makespan $M$ such that $F(x^M) = M + m^*$ for every $M \in \{|P| + m^* - m^*_P, \dots, |J| - m^*\}$.
    The case of $M = |P| + m^* - m^*_P$ follows from the previous case.
    \\
    Suppose that the property holds for some $M \in \{|P| + m^* - m^*_P, \dots, |J| - m^* - 1\}$ and let $x^M$ be the associated schedule.
    Using the induction hypothesis, $F(x^M) = M + m^*$ and, since $M < |J| - m^*$, $F(x^M) < |J|$.
    There must exist a time step with at least three jobs scheduled in $x^M$.
    A schedule $x^{M+1}$ with makespan $M+1$ and such that $F(x^{M+1}) = M + 1 + m^*$ can be obtained using one-job elongation as shown in Figure \ref{fig:simple_elongation_single_job}.
    This value reaches the upper bound of Lemma \ref{lemma:upper_bound_optimal_objective_function_2} so $x^{M+1}$ is optimal.

    For $M \geq |J| - m^*$, the schedule obtained through the method of Fujii \textit{et al.} is always optimal since it reaches the natural upper bound $|J|$ of the objective function.
\end{proof}

\subsection{Algorithm}

The main steps in the resolution of $L2 | prec, C_{max} \leq M, c_i = 1, p_i = 1 | F$ are summarized in Algorithm \ref{alg:L2_prec_UET}.
Schedule computation steps are not detailed but rely in particular on the proofs of Lemma \ref{lemma:augmenting_sequence} and Lemma \ref{lemma:generalized_augmenting_path}.

\begin{algorithm}
    \caption{Sketch of algorithm for $L2 | prec, C_{max} \leq M, c_i = 1, p_i = 1 | F$}
    \label{alg:L2_prec_UET}
    \begin{algorithmic}[1]
        \State Choose a critical path $P$ in $G$
        \State Compute $m^*_P$ and $m^*$ the size of a maximum matching in $\widetilde{G}_P$ and $\widetilde{G}$ respectively
        \If{$M \geq |J| - m^*$}
            \State \textbf{Return} the schedule computed with the method of \cite{Coffman1971}
        \EndIf
        \State Compute $x$ the earliest schedule of $G$
        \While{$F(x) < |P| + m^*_P$}
            \State Compute a schedule $x'$ of makespan $|P|$ with $F(x') = F(x)+1$ using Lemma \ref{lemma:augmenting_sequence}
            \State $x = x'$
        \EndWhile
        \If{$M = |P|$}
            \State \textbf{Return} $x$
        \EndIf
        \While{$C_{max}(x) < M$}
            \If{At least four jobs are scheduled at a time step in $x$}
                \State Compute a schedule $x'$ of makespan $C_{max}(x)+1$ with $F(x') = F(x)+2$ using a two-job elongation
                \State $x = x'$
            \ElsIf{$C_{max}(x) < |P| + m^* - m^*_P$}
                \State Compute a schedule $x'$ of makespan $C_{max}(x)+1$ with $F(x') = F(x)+2$ using Lemma \ref{lemma:generalized_augmenting_path}
                \State $x = x'$
            \Else
                \State Compute a schedule $x'$ of makespan $C_{max}(x)+1$ with $F(x') = F(x)+1$ using a one-job elongation
                \State $x = x'$
            \EndIf
        \EndWhile
        \State \textbf{Return} $x$
    \end{algorithmic}
\end{algorithm}

    Let us discuss the time complexity of Algorithm \ref{alg:L2_prec_UET}.
    An important point to note is that, although the method of \cite{Fujii1969} has been used, due to its convenient matching structure, to prove that $L2 | prec, C_{max} \leq M, c_i = 1, p_i = 1 | F$ is polynomial, it is not the most computationally efficient.
    Therefore, solving problem $P2 | prec, p_i = 1 | C_{max}$ when applying Lemma \ref{lemma:generalized_augmenting_path} is done using the more efficient method of \cite{Coffman1971}.
    The time required to perform the different computation steps is as follows:
    \begin{itemize}
        \item[-] Computing the transitive closure of precedence graph $G$: $O(|J|^3)$
        \item[-] Solving the problem for $M = |P|$:
        \begin{itemize}
            \item[-] Finding an augmenting sequence (alternating path in bipartite independence graph $\widetilde{G}_P$): $O(|J| + |E_P|)$
            \item[-] Elementary operation: $O(|J|)$
            \item[-] Size of an augmenting path: $O(|P|)$
            \item[-] Maximum number of augmenting sequences to find: $|P|$
        \end{itemize}
        Total to solve the $M = |P|$ case: $O(|P|^2|J|)$
        \item[-] Improving the schedule up to $M$:
        \begin{itemize}
            \item[-] Finding an augmenting sequence in an arbitrary graph: $O(|E|)$ \citep{Micali1980}
            \item[-] Rescheduling on two machines: $O(|J|^2)$ \citep{Coffman1971}
            \item[-] Maximum number of augmenting sequences to find: $|J|$
        \end{itemize}
        Total to improve the schedule up to $M$: $O(|J|^3)$
    \end{itemize}
    Total computation time of Algorithm \ref{alg:L2_prec_UET} is therefore $O(|J|^3)$.


\section{Further tractable special cases} \label{section:other_polynomial_cases}

In this section, five resource leveling problems are shown to be solvable in polynomial or pseudo-polynomial time.
Section \ref{section:NP_hardness_results} will show that further generalization of those special cases leads to strongly $NP$-hard problems.

Precedence constraints, as considered in the core problem of this work, are classical scheduling constraints and are present in many practical applications.
However, $L = 2$ is very specific and solving problems with other values of $L$ can be interesting.
Section \ref{section:L_intree} shows that, when precedence graphs are restricted to in-trees, the problem is solved to optimality for any $L \in \mathbb{N}$ by adapting Hu's algorithm \citep{Hu1961}.
Unit processing times is also a strong assumption and results for more general $p_i$ values can be of use.
Section \ref{section:L1_Cmax_cigeq1} shows that restricting the resource level to $L = 1$ yields another polynomial special case for any processing times.

Release and due dates are other classical scheduling constraints that can be studied together with a resource leveling objective.
Polynomial solving methods are given for two problems with release dates and due dates in which preemption is allowed.
Section \ref{section:L_ri_di_ci1_pmtn} deals with $L | r_i, d_i, c_i = 1, pmtn | F$ by translating it as a flow problem.
Section \ref{section:L2_ri_di_pmtn_F} solves $L2 | r_i, d_i, pmtn | F$ using linear programming.

Finally, Section \ref{section:L2_Cmax_F} shows that, when there are no precedence constraints, a tractable method exists for non-unit processing times and resource consumptions.
A pseudo-polynomial algorithm is provided for $L2 | C_{max} \leq M | F$ based on dynamic programming.

\subsection{A polynomial algorithm for $L | in \text{-} tree, C_{max} \leq M, p_i=1, c_i=1 | F$} \label{section:L_intree}

The classical scheduling problem $P | in \text{-} tree, C_{max} \leq M, p_i=1 | C_{max}$ can be solved in polynomial time using Hu's algorithm \citep{Hu1961}.
In this section, its leveling counterpart is considered, namely $L | in \text{-} tree, C_{max} \leq M, p_i=1, c_i=1 | F$.
The idea is therefore to adapt the algorithm proposed in \cite{Hu1961} in order to obtain a polynomial algorithm to solve $L | in \text{-} tree, C_{max} \leq M, p_i=1, c_i=1 | F$.
\\
Hu's algorithm can be applied in the case of Unit Execution Times (UET) when the precedence graph is an in-tree.
It is a list algorithm in which jobs are sorted in increasing order of latest starting time.
\\
The following algorithm uses the same priority list: jobs with small latest starting time $\bar{\tau}_i$ are scheduled first.
However, since the jobs are constrained by the deadline $M$, the algorithm must execute at a given time $\tau$ the jobs whose latest starting time is $\tau$ and that have not been scheduled yet.
This can cause the resource level $L$ to be exceeded.

\begin{algorithm}
    \caption{UETInTreeLeveling$(J, G, L, M)$}
    \label{alg:adapted_hu_algorithm}
    \begin{algorithmic}[1]
        \State Compute $\bar{\tau}_i$, the latest starting time of job $i$, for every $i \in J$
        \State $G_0 \gets G$
        \For{$\tau = 0, \dots, M-1$} \label{alg_line:adapted_hu_algorithm_main_loop}
            \State Select up to $L$ leaves of $G_\tau$ with the smallest values of $\bar{\tau}_i$, let this subset be $S_\tau$ \label{alg_line:adapted_hu_algorithm_job_selection}
            \Comment{Note that the selected leaves may have different $\bar{\tau}_i$ values}
            \State $S_\tau = S_\tau \cup \{i \text{ node of } G_\tau \mid \bar{\tau}_i = \tau \}$ \label{alg_line:adapted_hu_algorithm_special_condition}
            \ForAll{$i \in S_\tau$}
                \State $x_i \gets \tau$
            \EndFor
            \State Build $G_{\tau +1}$ from $G_\tau$ by removing the jobs of $S_\tau$
        \EndFor
        \State \textbf{Return} $x$
    \end{algorithmic}
\end{algorithm}

In order for a feasible schedule to exist, the deadline $M$ chosen as input for UETInTreeLeveling should be at least the length of the critical path in the precedence graph -- which is the depth of the in-tree in the present case.
Furthermore, it can be assumed that $M \leq |J|$, which ensures that UETInTreeLeveling is executed in polynomial time.
\\
UETInTreeLeveling returns schedules with a very specific structure.
The following lemma gives some of those specificities that will prove to be useful in showing that UETInTreeLeveling is actually optimal.

\begin{lemma} \label{lemma:structural_specificities_adapted_hu_algorithm}
    Let $x$ be the schedule returned by UETInTreeLeveling.
    \\
    Suppose that at a given time step $\tau \in \{0, \dots, M-1\}$ the set $S_\tau$ of jobs scheduled at $\tau$ in $x$ is such that $|S_\tau| > L$.
    The following assertions are true:
    \begin{itemize}
        \item[] (i) Every $i \in S_\tau$ has latest starting time $\bar{\tau}_i = \tau$;
        \item[] (ii) For every $\tau' \leq \tau$, $|S_{\tau'}| \geq L$;
        \item[] (iii) For every $\tau' \leq \tau$, every $i \in S_{\tau'}$ has latest starting time $\bar{\tau}_i \leq \tau$.
    \end{itemize}
\end{lemma}

\begin{proof}
    First, it is clear that when time step $\tau$ is reached in UETInTreeLeveling, all jobs with latest starting time at most $\tau - 1$ have already been scheduled.
    This is true thanks to the instruction of line \ref{alg_line:adapted_hu_algorithm_special_condition} that ensures that all jobs reaching their latest starting time are scheduled.
    \\
    As a consequence, all jobs with latest starting time $\tau$ that remain in $G_\tau$ are leaves and they have the lowest $\bar{\tau}_i$ value.
    There are at least $L + 1$ such jobs otherwise it would be impossible to have $|S_{\tau}| > L$.
    The instruction of line \ref{alg_line:adapted_hu_algorithm_job_selection} therefore selects $L$ jobs with latest starting time $\tau$ and instruction of line \ref{alg_line:adapted_hu_algorithm_special_condition} adds the remaining ones to $S_\tau$. 
    This proves assertion (i).
    \\
    Due to assertion (i) and since $|S_\tau| > L$, there are at least $L$ jobs with latest starting time $\tau$.
    For any $\tau' \leq \tau$, $G_\tau$ is a subtree of $G_{\tau'}$ so $G_{\tau'}$ has at least $L$ leaves corresponding to jobs with latest starting time at most $\tau$.
    In particular, $G_{\tau'}$ has at least $L$ leaves, which ensures that $|S_{\tau'}| \geq L$ and proves assertion (ii).
    \\
    Since UETInTreeLeveling prioritizes jobs with lower $\bar{\tau}_i$ values and since for any $\tau' \leq \tau$, $G_{\tau'}$ has at least $L$ leaves corresponding to jobs with latest starting time $\tau$ or less, no job with latest starting time strictly higher than $\tau$ can be scheduled at $\tau$ or before.
    This proves assertion (iii).
\end{proof}

\begin{proposition} \label{prop:optimality_adapted_hu_algorithm}
    $L | in \text{-} tree, C_{max} \leq M, p_i=1, c_i=1 | F$ is solvable in polynomial time using UETInTreeLeveling.
\end{proposition}

\begin{proof}
    Let $x$ be the schedule returned by UETInTreeLeveling.
    Figure \ref{fig:UETInTreeLeveling_schedule} gives an example of such a schedule and illustrates the decomposition of the objective function that follows.
    For any time step $\tau \in \{0, \dots, M-1\}$, let $S_{\tau}$ denote the set of jobs scheduled at $\tau$ in $x$.
    \\
    First, $x$ is a feasible schedule.
    Indeed, by selecting leaves of the subtree $G_\tau$, the instruction of line \ref{alg_line:adapted_hu_algorithm_job_selection} ensures that the predecessors of the selected jobs have already been scheduled, thus satisfying precedence constraints.
    As for the deadline constraint, it is satisfied thanks to the instruction of line \ref{alg_line:adapted_hu_algorithm_special_condition} that ensures that each job is scheduled no later than its latest starting time.
    \\
    If for every $\tau \in \{0, \dots, M-1\}$, $|S_\tau| \leq L$, then $x$ is optimal since $F(x) = |J|$, which reaches a natural upper bound on $F$.
    \\
    If not, let $\tau$ be the largest time step such that $|S_\tau| > L$.
    In any feasible schedule, all jobs $i$ with latest starting time $\bar{\tau}_i \leq \tau$ must be scheduled at $\tau$ or before.
    It implies that only jobs $i$ with latest starting time $\bar{\tau}_i > \tau$ can be scheduled on time steps $\tau+1, \tau+2, \dots, M-1$.
    An upper bound on $F$ is therefore given by $(\tau+1)L + |\{i \in J \mid \bar{\tau}_i > \tau\}|$.
    \\
    Assertion (ii) of Lemma \ref{lemma:structural_specificities_adapted_hu_algorithm} gives that the contribution to $F(x)$ of interval $[0, \tau+1]$ is $\sum_{\tau' = 0}^{\tau} \min(L, |S_{\tau'}|) = (\tau+1)L$.
    Since $\tau$ is the last time step such that $|S_\tau| > L$ and given assertion (iii) of Lemma \ref{lemma:structural_specificities_adapted_hu_algorithm}, all jobs of $\{i \in J \mid \bar{\tau}_i > \tau\}$ are scheduled on time steps $\tau+1, \tau+2, \dots, M-1$ without exceeding resource level $L$.
    It then holds that the contribution to $F(x)$ of interval $[\tau+1, M]$ is $|\{i \in J \mid \bar{\tau}_i > \tau\}|$.
    Finally, the value of $F(x)$ writes:
    $$F(x) = (\tau+1)L + |\{i \in J \mid \bar{\tau}_i > \tau\}|$$
    which reaches the previously given upper bound and is therefore optimal.
\end{proof}

\begin{figure*}[h!]
    \centering
    \scalebox{0.9}{\includegraphics{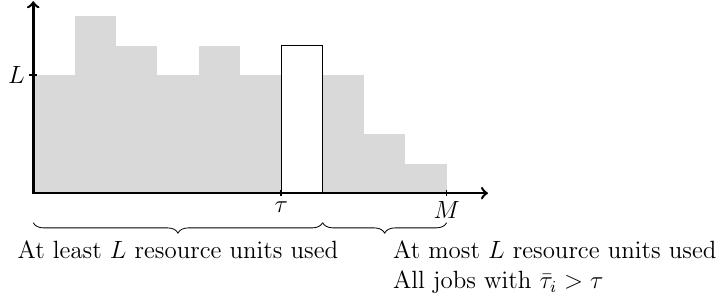}}
    \caption{Example of schedule resulting from UETInTreeLeveling}
    \label{fig:UETInTreeLeveling_schedule}
\end{figure*}

\subsection{A polynomial algorithm for $L1 | prec, C_{max} \leq M, c_i > 0 | F$} \label{section:L1_Cmax_cigeq1}

Suppose that $L = 1$ and that $c_i > 0$ for every $i \in J$.
Let us also assume that the makespan deadline $M$ is larger than the longest path in the precedence graph -- otherwise no feasible schedule exists.
For any $i \in J$, let us denote $\bar{\tau}_i$ the latest starting time of job $i$, i.e., the makespan deadline $M$ minus the longest path from job $i$ in the precedence graph $G$.
\\
Algorithm UnitResourceLeveling below starts by scheduling jobs sequentially in a topological order and then schedules the remaining ones at their latest starting time.

\begin{algorithm}
    \caption{UnitResourceLeveling$(J, G, p)$}
    \label{alg:unit_resource_level}
    \begin{algorithmic}[1]
        \State Compute $\bar{\tau}_i$, the latest starting time of job $i$, for every $i \in J$
        \State Choose a topological order $i_1, i_2, \dots, i_n$ on $J$
        \State $x_{i_1} = 0$
        \ForAll{$k \in \{2, \dots, n\}$} \label{alg_line:unit_resource_level_main_loop}
            \State $x_{i_k} \gets \min(\sum_{k' = 1}^{k - 1}p_{i_{k'}}, \bar{\tau}_{i_k})$ \label{alg_line:unit_resource_level_value_choice}
        \EndFor
        \State \textbf{Return} $x$
    \end{algorithmic}
\end{algorithm}

\begin{proposition} \label{prop:unit_resource_leveling_correctness}
    UnitResourceLeveling solves $L1 | prec, C_{max} \leq M, c_i > 0 | F$ in polynomial time.
\end{proposition}

\begin{proof}
    Let us prove that UnitResourceLeveling can be executed in $O(|J|+|\mathcal{A}|)$ time (recall that $\mathcal{A}$ is the set of arcs in the precedence graph).
    First, $\bar{\tau}_i$ values for every $i \in J$ as well as the topological order can be computed with a graph traversal in $O(|J|+|\mathcal{A}|)$ time.
    As for the main loop, it can be executed in $O(|J|)$ time.
    \\
    Let us now prove that UnitResourceLeveling provides a schedule that maximizes the time during which at least one unit of resource is used.
    \\
    In the particular case where $\sum_{i \in J} p_i \leq M$, it is easily seen that all the jobs can be scheduled consecutively in a topological order thus leading to a resource usage that is always at most $L = 1$.
    It is exactly what the algorithm does since, when $\sum_{i \in J} p_i \leq M$, the inequality $\sum_{k' = 1}^{k - 1}p_{i_k} \leq \bar{\tau}_{i_k}$ always holds on line \ref{alg_line:unit_resource_level_value_choice}.
    \\
    If $\sum_{i \in J} p_i > M$ then there exists $k>1$ such that $\sum_{k' = 1}^{k - 1}p_{i_{k'}} > \bar{\tau}_{i_k}$.
    Let $k$ be the first such index.
    In the schedule returned by UnitResourceLeveling, jobs $i_1, i_2, \dots, i_{k-1}$ are scheduled consecutively without interruption on time interval $[0, \sum_{k' = 1}^{k - 1}p_{i_{k'}}]$.
    By definition of $\bar{\tau}_{i_k}$, the jobs on a longest path starting from $i_k$ are executed without interruption on time interval $[\bar{\tau}_{i_k}, M]$.
    Given that $\sum_{k' = 1}^{k - 1}p_{i_{k'}} > \bar{\tau}_{i_k}$, the returned schedule uses at least one unit of resource constantly on time interval $[0, M]$.
    The amount of resource below the level $L = 1$ is therefore maximized.
    \\
    Also note that the vector $x$ returned by UnitResourceLeveling is a feasible schedule as it satisfies precedence constraints.
    Indeed, given an arc $(i, j)$ of the precedence graph it is clear that, by definition, $\bar{\tau}_i + p_i \leq \bar{\tau}_j$ and that $i$ comes before $j$ in any topological order.
    If $x_j = \bar{\tau}_j$, then $x_i \leq \bar{\tau}_i \leq \bar{\tau}_j - p_i = x_j - p_i$ and therefore $x_i + p_i \leq x_j$.
    Assuming that $i = i_k$ and $j = i_{k'}$ in the topological order, if $x_j = \sum_{k'' = 1}^{k' - 1}p_{i_{k''}}$, then $x_i \leq \sum_{k'' = 1}^{k - 1}p_{i_{k''}} \leq \sum_{k'' = 1}^{k' - 1}p_{i_{k''}} - p_i = x_j - p_i$ and therefore $x_i + p_i \leq x_j$.
\end{proof}

\subsection{A polynomial algorithm for $L | r_i, d_i, c_i = 1, pmtn | F$} \label{section:L_ri_di_ci1_pmtn}
The aim of this section is to show that Problem $L | r_i, d_i, c_i = 1, pmtn | F$ reduces to a linear program.
Let $\tau_0 < \tau_1 < \tau_2 < \dots < \tau_K$ be such that $\{\tau_0, \tau_1, \dots, \tau_K\} = \{r_i \mid i \in J\} \cup \{d_i \mid i \in J\}$.
For convenience, let also $\kappa_i = \{k \in \{1, \dots, K\}, [\tau_{k-1}, \tau_k] \subseteq [r_i, d_i]\}$.
The following variables are considered:
\begin{itemize}
    \item[-] $x_{ik}$ $\forall i \in J,\; \forall k \in \kappa_i$: the number of time units of job $i$ processed during interval $[\tau_{k-1}, \tau_k]$;
    \item[-] $y_k$ $\forall k \in \{1, \dots, K\}$: number of time units processed on interval $[\tau_{k-1}, \tau_k]$ not exceeding the resource level;
    \item[-] $z_k$ $\forall k \in \{1, \dots, K\}$: number of time units processed on interval $[\tau_{k-1}, \tau_k]$ over the resource level;
\end{itemize}
Note that the following program only provides an implicit solution of Problem $L | r_i, d_i, c_i = 1, pmtn | F$.
Variables $x_{ik}$ give the number of time units of a job $i$ in the interval $[\tau_{k-1}, \tau_k]$ but this does not tell when each individual time unit is executed.
It will be shown later that an algorithm can be used to deduce a fully-fledged solution of $L | r_i, d_i, c_i = 1, pmtn | F$ from $x_{ik}$ values.
\\
The program writes: 
\begin{align}
    \max_{x, y, z} \quad &\sum_{k = 1}^{K} y_k \nonumber\\
    s.t. \quad &
        \sum_{k \in \kappa_i} x_{ik} = p_i
        &\forall i \in J \label{constraint:first}
    \\
    &
        y_k + z_k = \sum_{\substack{i \in J \\ \kappa_i \ni k}} x_{ik}
        &\forall k \in \{1, \dots, K\}
    \\
    &
        0 \leq x_{ik} \leq \tau_k - \tau_{k-1}
        &\forall i \in J,\; \forall k \in \{1, \dots, K\}
    \\
    &
        0 \leq y_k \leq (\tau_k - \tau_{k-1})L
        &\forall k \in \{1, \dots, K\}
    \\
    &
        z_k \geq 0
        &\forall k \in \{1, \dots, K\} \label{constraint:last}
\end{align}
Note that in this program, maximizing $\sum_{k = 1}^{K} y_k$ is equivalent to minimizing $\sum_{k = 1}^{K} z_k$ since $\sum_{k = 1}^{K} y_k + z_k = \sum_{i \in J} p_i$.
The alternative program writes:
\begin{align*}
    \min_{x, y, z} \quad &\sum_{k = 1}^{K} z_k \nonumber\\
    s.t. \quad &(\ref{constraint:first}) \dots (\ref{constraint:last})
\end{align*}
It can be interpreted as a minimum cost flow problem.
Figure \ref{fig:flow_graph} shows the structure of the flow graph with flow values on arcs.
The costs are equal to one for arcs with $z_k$ and $0$ otherwise.
Arcs with $x_{ik}$ have capacity $\tau_{k} - \tau_{k-1}$, arcs with $y_{k}$ have capacity $L(\tau_{k} - \tau_{k-1})$ and arcs with $z_{k}$ have infinite capacity.

\begin{figure*}[h!]
    \centering
    \scalebox{0.9}{\includegraphics{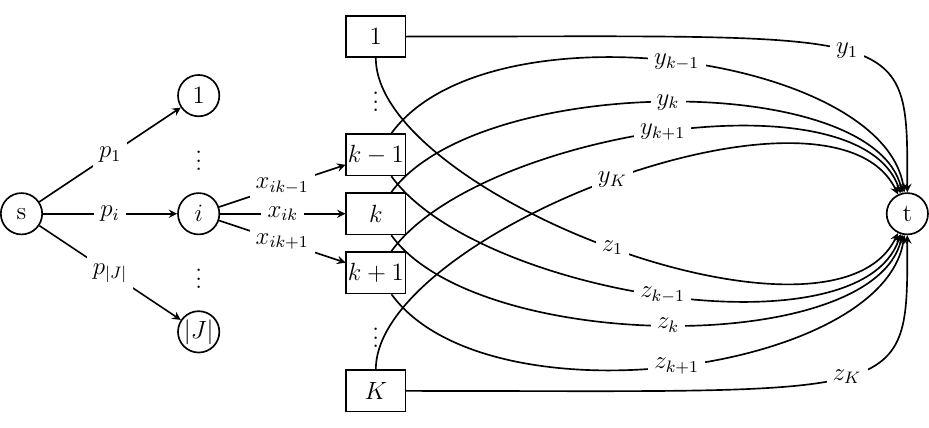}}
    \caption{Flow graph with source $s$ and sink $t$}
    \label{fig:flow_graph}
\end{figure*}

An explicit solution of Problem $L | r_i, d_i, c_i = 1, pmtn | F$ can be constructed with an additional post-processing step.
Consider the interval $[\tau_{k-1}, \tau_k]$ for some $k$ in $\{1, \dots, K\}$ and let $J_k = \{i \in J \mid [\tau_{k-1}, \tau_k] \subseteq [r_i, d_i]\}$.
\\
The following algorithm, illustrated in Figure \ref{fig:interval_scheduling}, computes the sub-intervals of $[\tau_{k-1}, \tau_k]$ on which each job of $J_k$ is executed.
Its idea is simple: the interval is filled line by line from left to right by placing one job after another and starting a new line every time the end of the interval is reached.
Those line breaks may cause a job to be split into two parts, one reaching the end of the interval and the other starting back from the beginning.

\begin{algorithm}
    \caption{IntervalScheduling$(J_k, x_{.k}, \tau_{k-1}, \tau_k)$}
    \label{alg:interval_job_scheduling}
    \begin{algorithmic}[1]
        \State $t \gets 0$
        \ForAll{$i \in J_k$} \label{alg_line:interval_job_scheduling_main_loop}
            \State $t' = t + x_{i k}$
            \If{$t' \leq \tau_k - \tau_{k-1}$}
                \State $I_i \gets \{[\tau_{k-1} + t, \tau_{k-1} + t']\}$
            \Else
                \State $t' = t' - (\tau_k - \tau_{k-1})$
                \State $I_i = \{[\tau_{k-1}, \tau_{k-1} + t'], [\tau_{k-1} + t, \tau_k]\}$
            \EndIf
            \State $t = t'$
        \EndFor
        \State \textbf{Return} $(I_i)_{i \in J_k}$
    \end{algorithmic}
\end{algorithm}

\begin{figure}[h!]
    \centering
    \includegraphics{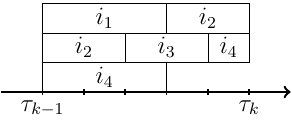}
    \caption{Example of interval job scheduling for $J_k = \{i_1, i_2, i_3, i_4\}$, $\tau_k - \tau_{k-1} = 5$, $x_{i_1 k} = 3$, $x_{i_2 k} = 4$, $x_{i_3 k} = 2$, $x_{i_4 k} = 4$}
    \label{fig:interval_scheduling}
\end{figure}

It is clear that IntervalScheduling is executed in $O(|J|)$ time.
Executing it for each $k \in \{1, \dots, K\}$ can therefore be done in $O(|J|^2)$.
Note that in the solution provided by IntervalScheduling, each job $i_l$ is executed on either one or two sub-intervals of $[\tau_{k-1}, \tau_k)$.
As a consequence, each job is executed on at most $O(|J|)$ disjoint intervals in the final solution. 

\begin{proposition} \label{prop:premptive_interval_job_scheduling}
    $L | r_i, d_i, c_i = 1, pmtn | F$ is solvable in polynomial time.
\end{proposition}

\begin{proposition} \label{prop:UET_interval_job_scheduling}
    $L | r_i, d_i, p_i = 1, c_i = 1 | F$ is solvable in polynomial time.
\end{proposition}

\begin{proof}
    Since $L | r_i, d_i, c_i = 1, pmtn | F$ reduces to a minimum cost flow problem, standard algorithms can be used to solve it and yield integer solutions when $p_i, r_i, d_i \in \mathbb{N}$.
    In particular, optimal integer solutions can be found for $L | r_i, d_i, p_i = 1, c_i = 1, pmtn | F$, which is actually the same as finding optimal solutions for $L | r_i, d_i, p_i = 1, c_i = 1 | F$.
\end{proof}

\subsection{A polynomial algorithm for $L2 | r_i, d_i, pmtn | F$} \label{section:L2_ri_di_pmtn_F}

Similarly as in Section \ref{section:L_ri_di_ci1_pmtn}, the aim of this section is to show that problem $L2 | r_i, d_i, pmtn | F$ reduces to a linear program.
It will be assumed that all considered instances are feasible, namely that $p_i \leq d_i - r_i$ for every job $i \in J$.
\\
In Problem $L2 | r_i, d_i, pmtn | F$, the resource consumption $c_i$ of job $i \in J$ can take any value in $\mathbb{N}$, yet it is sufficient to consider the case where $c_i \in \{1, 2\}$.
Indeed, since objective function $F$ only takes into consideration the resource consumption that fits under resource level $L = 2$, any $c_i > 2$ can be handled as $c_i = 2$ without changing the optimum.
As for potential jobs $c_i = 0$, they have no impact on the objective.

Let us then suppose that $c_i \in \{1, 2\}$ for every $i \in J$ and denote $J_1 = \{i \in J \mid c_i = 1\}$ and $J_2 = \{i \in J \mid c_i = 2\}$.
Let $\tau_0 < \tau_1 < \tau_2 < \dots < \tau_K$ be such that $\{\tau_0, \tau_1, \dots, \tau_K\} = \{r_i \mid i \in J\} \cup \{d_i \mid i \in J\}$.
For convenience, let also $\kappa_i = \{k \in \{1, \dots, K\}, [\tau_{k-1}, \tau_k] \subseteq [r_i, d_i]\}$.
The following variables are considered:
\begin{itemize}
    \item[-] $x^1_{ik}$ $\forall i \in J_1,\; \forall k \in \kappa_i$: the number of time units of job $i$ processed during interval $[\tau_{k-1}, \tau_k]$ that fit under the resource level;
    \item[-] $x^2_{ik}$ $\forall i \in J_2,\; \forall k \in \kappa_i$: the number of time units of job $i$ processed during interval $[\tau_{k-1}, \tau_k]$ that fit under the resource level.
\end{itemize}

Just like in Section \ref{section:L_ri_di_ci1_pmtn}, the following linear program only gives an implicit solution of the problem from which a proper schedule can then be deduced.
The program writes:
\begin{align}
    \max_{x^1, x^2} \quad &\sum_{i \in J_1} \sum_{k \in \kappa_i} x^1_{ik} + 2 \sum_{i \in J_2} \sum_{k \in \kappa_i} x^2_{ik} \nonumber\\
    s.t. \quad &
        \sum_{k \in \kappa_i} x^\alpha_{ik} \leq p_i
        &\forall \alpha \in \{1, 2\},\; \forall i \in J_\alpha
    \\
    &
        x^1_{ik} \leq (\tau_{k+1} - \tau_k) - \sum_{\substack{j \in J_2\\\kappa_j \ni k}} x^2_{jk}
        &\forall i \in J_1,\; \forall k \in \kappa_i
    \\
    &
        \sum_{\substack{i \in J_1\\\kappa_i \ni k}} x^1_{ik} \leq 2\left((\tau_{k+1} - \tau_k) - \sum_{\substack{j \in J_2\\\kappa_j \ni k}} x^2_{jk}\right)
        &\forall k \in \{1, \dots, K\}
    \\
    &
        \sum_{\substack{i \in J_2\\\kappa_i \ni k}} x^2_{ik} \leq (\tau_{k+1} - \tau_k)
        \forall k \in \{1, \dots, K\}
    \\
    &
        x^\alpha_{ik} \geq 0
        &\forall \alpha \in \{1, 2\},\; \forall i \in J_\alpha,\; \forall k \in \kappa_i
\end{align}

A complete solution in which each job is given a set of intervals on which it is executed can be deduced from optimal values of $x^1_{ik}$ and $x^2_{ik}$ variables.
Recall that variable $x^1_{ik}$ and $x^2_{ik}$ only represent the portions of the jobs that fit under the resource level.
Those portions of jobs are scheduled on each interval $[\tau_{k-1}, \tau_k]$ as shown in Figure \ref{fig:L2_interval_scheduling}: the jobs of $J_2$ are scheduled first, then the jobs of $J_1$ are scheduled in the remaining space with Algorithm \ref{alg:interval_job_scheduling}.
Once this is done, the processing time left for each job $i$ can be scheduled anywhere in the remaining space of its availability interval $[r_i, d_i]$.
The initial assumption that $p_i \leq d_i - r_i$ for every job $i \in J$ ensures that it is possible to do so.

\begin{figure}[h!]
    \centering
    \includegraphics{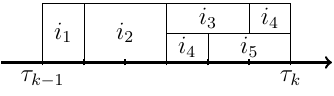}
    \caption{Example of interval job scheduling}
    \label{fig:L2_interval_scheduling}
\end{figure}

\begin{proposition} \label{prop:L2_preemptive_interval_job_scheduling}
    $L2 | r_i, d_i, pmtn | F$ is solvable in polynomial time.
\end{proposition}

\begin{remark}
    Unlike in the case of Section \ref{section:L_ri_di_ci1_pmtn}, nothing guarantees to find integer solution to the problem.
\end{remark}

\begin{remark}
    The method can be generalized to Problem $L | r_i, d_i, pmtn, c_i \in \{1, L\} | F$, in which jobs either consume one or $L$ units of resource.
\end{remark}

\subsection{A pseudo-polynomial algorithm for $L2 | C_{max} \leq M | F$} \label{section:L2_Cmax_F}

In order to solve $L2 | C_{max} \leq M | F$ in pseudo-polynomial time, some assumptions are made w.l.o.g. to exclude trivial cases.
First, it can be assumed that resource consumptions values are in $\{0, 1, 2\}$.
Indeed, when $L = 2$, replacing resource consumption of each job $i \in J$ by $\widetilde{c}_i = \min(c_i, 2)$ does not change the objective function values.
Jobs with resource consumption $c_i = 0$ have no impact and can be scheduled at any feasible date, they can be ignored.
Only jobs with $c_i \in \{1, 2\}$ are to be considered in the sequel.

Let us denote $J_1 = \{i \in J \mid c_i = 1\}$ and $J_2 = \{i \in J \mid c_i = 2\}$.
Let also $p_{J_1} = \sum_{i \in J_1} p_i$ and $p_{J_2} = \sum_{i \in J_2} p_i$.
If $p_{J_2} \geq M$, the problem is solved simply by scheduling the jobs of $J_2$ one after another while possible and at their latest possible date afterwards.
This uses at least two units of resource constantly on $[0, M]$ thus maximizing $F$.
Let us therefore assume that $p_{J_2} < M$.

The following lemma links the objective value of a schedule $x$ to a subset of $J_1$, allowing for the construction of the former from the latter.

\begin{lemma} \label{lemma:L2_F_geq_subset_sum}
    Let $I = (J, p, c, M)$ be an instance of $L2 | C_{max} \leq M| F$.
    \\
    For any subset of jobs $J' \subseteq J_1$, there exists a schedule $x$ such that:
    \begin{align*}
        F(x) \geq 2 p_{J_2} &+ \min(\sum_{i \in J'} p_i, M - p_{J_2}) \\
        &+ \min(\sum_{i \in J_1 \setminus J'} p_i, M - p_{J_2})
    \end{align*}
    Furthermore, $x$ can be computed in polynomial time.
\end{lemma}

\begin{proof}
    The schedule $x$ is built as follows and as illustrated in Figure \ref{fig:L2_solution_construction}:
    \begin{itemize}
        \item[-] The jobs of $J_2$ are scheduled one after the other from $0$ to $p_{J_2}$.
        \item[-] The jobs of $J'$ are scheduled one after the other, starting from $p_{J_2}$, while the makespan deadline is not exceeded and at their latest possible date afterwards.
        \item[-] The jobs of $J_1 \setminus J'$ are scheduled as the jobs of $J'$ 
    \end{itemize} 
    \begin{figure*}[h!]
        \centering
        \scalebox{0.9}{\includegraphics{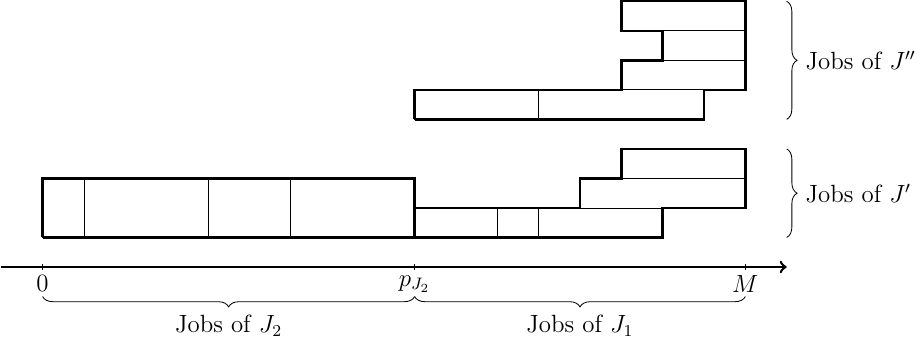}}
        \caption{Schedule construction}
        \label{fig:L2_solution_construction}
    \end{figure*}
    The schedule $x$ is such that:
    \begin{itemize}
        \item[-] The jobs of $J_2$ use exactly two units of resource on $[0, p_{J_2}]$;
        \item[-] The jobs of $J'$ use at least one unit of resource on $[p_{J_2}, p_{J_2} + \min(\sum_{i \in  J'} p_i, M - p_{J_2})]$;
        \item[-] The jobs of $J_1 \setminus J'$ use at least one unit of resource on $[p_{J_2}, p_{J_2} + \min(\sum_{i \in  J_1 \setminus J'} p_i, M - p_{J_2})]$.
    \end{itemize}
    This implies that $F(x) \geq 2 p_{J_2} + \min(\sum_{i \in J'} p_i, M - p_{J_2}) + \min(\sum_{i \in J_1 \setminus J'} p_i, M - p_{J_2})$.
\end{proof}

Having two jobs of $J_2$ that overlap in a schedule is never advantageous in terms of leveling objective.
This idea is formalized in the following lemma:

\begin{lemma} \label{lemma:L2_non_overlap_dominance}
    The schedules in which the jobs of $J_2$ do not overlap are dominant.
\end{lemma}

\begin{proof}
    One job of $J_2$ is sufficient to reach resource level $L = 2$ for its whole duration and overlapping with another job is a missed contribution to function $F$.
    Given the previous assumption that $p_{J_2} < M$, it is always possible to find an optimal schedule in which the jobs of $J_2$ do not overlap.
\end{proof}

The following lemma complements the inequality of Lemma \ref{lemma:L2_F_geq_subset_sum} by another inequality that is valid  when the jobs of $J_2$ do not overlap.

\begin{lemma} \label{lemma:L2_F_leq_subset_sum}
    Let $I = (J, p, c, M)$ be an instance of $L2 | C_{max} \leq M| F$ and let $x \in \mathbb{N}^J$ be a feasible schedule for $I$ such that the jobs of $J_2$ do not overlap.
    \\
    Then there exists a subset of jobs $J' \subseteq J_1$ such that:
    \begin{align*}
        F(x) \leq 2 p_{J_2} &+ \min(\sum_{i \in J'} p_i, M - p_{J_2}) \\
        &+ \min(\sum_{i \in J_1 \setminus J'} p_i, M - p_{J_2})
    \end{align*}
\end{lemma}

\begin{proof}
    Let us build two disjoint subsets $J', J'' \subseteq J_1$ from $x$ using the following procedure:
    \begin{algorithmic}[1]
        \State $J' \gets \emptyset$ $J'' \gets \emptyset$
        \For{$\tau = 0, \dots, M-1$} 
            \If{a job of $J_2$ is processed on time step $\tau$ in $x$}
                \State \textbf{Continue}
            \EndIf
            \For{$\widetilde{J} = J', J''$}
                \If{a job of $\widetilde{J}$ is processed on time step $\tau$ in $x$}
                    \State \textbf{Continue}
                \ElsIf{a job of $i \in J \setminus (J' \cup J'')$ is processed on time step $\tau$ in $x$}
                    \State $\widetilde{J} = \widetilde{J} \cup \{i\}$
                \EndIf
            \EndFor
        \EndFor
        \State \textbf{Return} $J', J''$
    \end{algorithmic}
    Let us denote $T'$ (resp. $T''$) the set of time steps where no job of $J_2$ is processed and where at least one job of $J'$ (resp. $J''$) is processed.
    It is clear that $|T'| \leq \min(\sum_{i \in J'} p_i, M - p_{J_2})$ and $|T''| \leq \min(\sum_{i \in J''} p_i, M - p_{J_2})$.
    For any time step $\tau \in \{0, \dots, M-1\}$, three cases are possible:
    \begin{itemize}
        \item[-] $\tau \in T' \cap T''$: at least two jobs of $J_1$ are processed at time step $\tau$ in $x$, one being in $J'$ and another in $J''$;
        \item[-] $\tau \in (T' \setminus T'') \cup (T'' \setminus T')$: exactly one job of $J_1$ is processed at time step $\tau$ in $x$, it is either in $J'$ or in $J''$;
        \item[-] $\tau \notin T' \cup T''$: either a job of $J_2$ is processed at time step $\tau$ in $x$ or no job at all, otherwise it would have been added to either $J'$ or $J''$ in the procedure.
    \end{itemize}
    In other words, $F(x) = 2 p_{J_2} + |T'| + |T''|$ and as a consequence 
    $F(x) \leq 2 p_{J_2} + \min(\sum_{i \in J'} p_i, M - p_{J_2}) + \min(\sum_{i \in J''} p_i, M - p_{J_2})$.
    The final inequality
    $F(x) \leq 2 p_{J_2} + \min(\sum_{i \in J'} p_i, M - p_{J_2}) + \min(\sum_{i \in J_1 \setminus J'} p_i, M - p_{J_2})$
    is derived using that $J'' \subseteq J_1 \setminus J'$.
\end{proof}

Thanks to Lemma \ref{lemma:L2_non_overlap_dominance}, the inequality of Lemma \ref{lemma:L2_F_leq_subset_sum} actually applies to any schedule.
Combining it with the inequality of Lemma \ref{lemma:L2_F_geq_subset_sum} yields the following equality:
\begin{align}
    \max_{x} F(x) = 2 p_{J_2} + \max_{J' \subseteq J_1} &\left(\min(\sum_{i \in J'} p_i, M - p_{J_2}) \right.\nonumber\\
    &\left. + \min(\sum_{i \in J_1 \setminus J'} p_i, M - p_{J_2})\right) \label{eqation:subset_sum_equivalence}
\end{align}
It is then clear that if a subset $J' \subseteq J_1$ is found that maximizes $\min(\sum_{i \in J'} p_i, M - p_{J_2}) + \min(\sum_{i \in J_1 \setminus J'} p_i, M - p_{J_2})$, the maximum value of $F$ over feasible schedules can be deduced directly.
Furthermore, using Lemma \ref{lemma:L2_F_geq_subset_sum}, it is actually possible to build a schedule $x$ from $J'$ such that:
\begin{align*}
    F(x) \geq 2 p_{J_2} &+ \min(\sum_{i \in J'} p_i, M - p_{J_2}) \\
    &+ \min(\sum_{i \in J_1 \setminus J'} p_i, M - p_{J_2})
\end{align*}
Schedule $x$ thus maximizes $F$, it is an optimal solution of the problem.

A possible interpretation of Equation \ref{eqation:subset_sum_equivalence} is that solving $L2 | C_{max} \leq M| F$ is equivalent to solving a two-machine early work maximization problem on the jobs of $J_1$: 
subset $J'$ (resp. $J \setminus J'$) corresponds to the jobs scheduled on the first (resp. second) machine and the portion of jobs executed before date $M - p_{J_2}$ is maximized.
By complementarity of the early and late work criterions, the equivalence with the late work minimization version of the problem also holds.
Late work minimization on two machines with common deadline, namely $P2 | d_j = d | Y$ is studied in \cite{Chen2016} and shown to be solvable in pseudo-polynomial time.
Applying the algorithm proposed in \cite{Chen2016} leads to a solution method for $L2 | C_{max} \leq M| F$ running in $O(|J| + |J_1|p_{j_1}^2)$.
Yet, finding an optimal subset $J'$ can be done with a better time complexity using standard dynamic programming.
A table $T[k, b]_{0 \leq k \leq |J_1|,\; 0 \leq b \leq p_{J_1}}$ can be computed in time $O(|J_1|p_{J_1})$ where $T[k, b] = 1$ if there exists a subset of the first $k$ jobs of $J_1$ whose sum is $b$ and $T[k, b] = 0$ otherwise.
Then, for each $b$ such that $T[|J_1|, b] = 1$,  the values of $\min(b, M - p_{J_2}) + \min(p_{J_1} - b, M - p_{J_2})$ can be computed in order to find the maximum value.
This maximum value is optimal and the corresponding optimal subset $J' \subseteq J_1$ can be recovered with a standard backward technique in time $O(|J_1|)$.

This leads to the following proposition:

\begin{proposition} \label{prop:L2}
    Problem $L2 | C_{max} \leq M| F$ can be solved in $O(|J| + |J_1|p_{J_1})$ time.
\end{proposition}

\begin{proof}
    A subset $J' \subseteq J_1$ that maximizes $\min(\sum_{i \in J'} p_i, M') + \min(\sum_{i \in J_1 \setminus J'} p_i, M')$ can be built in time $O(|J_1|p_{J_1})$ using dynamic programming.
    \\
    Lemmas \ref{lemma:L2_F_geq_subset_sum} and \ref{lemma:L2_F_leq_subset_sum} guarantee that $J'$ satisfies $\max_{x} F(x)  = 2 p_{J_2} + \min(\sum_{i \in J'} p_i, M') + \min(\sum_{i \in J_1 \setminus J'} p_i, M')$.
    Moreover, lemma \ref{lemma:L2_F_geq_subset_sum} states that a schedule $x$ can be built from $J'$ that maximizes $F$.
    This can be done in $O(|J|)$ time.
    \\
    An optimal schedule for an instance of problem $L2 | C_{max} \leq M, c_i = 1| F$ can therefore be found with a total computation time of $O(|J| + |J_1|p_{J_1})$.
\end{proof}

\section{$NP$-hardness results} \label{section:NP_hardness_results}

Some of the algorithms presented in the previous sections highlight similarities between resource leveling and classical scheduling problems in terms of solution.
This section shows that hardness results easily transfer from machine scheduling problems to their leveling counterparts.
This will follow from the Lemma~\ref{lemma:general_reduction}; it intuitively shows that resource leveling problems can be seen as generalizations of classical machine scheduling problems.

\begin{lemma} \label{lemma:general_reduction}
    Let $I'$ be an instance of $Pk | \beta | C_{max}$ for some $k \in \mathbb{N}$ and some set of constraints $\beta$.
    Let $I'' = (I', M)$ be an instance of $Lk | \beta, C_{max} \leq M, c_i = 1| F$.
    \\
    The two following assertions are equivalent:
    \begin{itemize}
        \item[] (i) $I'$ has a solution of makespan at most $M$
        \item[] (ii) $I''$ has a feasible schedule $x$ with $F(x) = \sum_{i \in J} p_i$
    \end{itemize}
\end{lemma}

\begin{proof}
    First note that a solution of $Pk | \beta | C_{max}$ is described by a schedule $x$ and a partition $(J_1, J_2, \dots, J_k)$ of $J$ such that $J_l$ is the subset of non-overlapping jobs processed on the $l$-th machine.
    \\
    (i) $\Rightarrow$ (ii) Let then $(x, (J_1, J_2, \dots, J_k))$ be a feasible solution of $Pk | \beta | C_{max}$ with makespan at most $M$.
    It is clear that the schedule $x$ is a feasible schedule for $Lk | \beta, C_{max} \leq M, c_i = 1 | F$ since it satisfies the constraints of $\beta$ and has makespan at most $M$.
    Furthermore, the jobs can be processed on $k$ machines, which implies that the resource consumption never exceeds $L = k$ and gives that $F(x) = \sum_{i \in J} p_i$.
    As a consequence, $Lk | \beta, C_{max} \leq M, c_i = 1 | F$ has a feasible schedule $x$ with $F(x) = \sum_{i \in J} p_i$.
    \\
    (i) $\Leftarrow$ (ii) Let $x$ be a feasible schedule for $Lk | \beta, C_{max} \leq M, c_i = 1| F$ such that $F(x) = \sum_{i \in J} p_i$.
    Using the following procedure, let us build a partition $J_1, J_2, \dots, J_k$ of $J$ such that each $J_l$ is a subset of non-overlapping jobs.
    \begin{algorithmic}[1]
        \State $J_l \gets \emptyset$ $\forall l \in \{1, \dots, k\}$
        \State Sort the jobs of $J$ in increasing order of starting time in $x$
        \\ Let $i_1, i_2, \dots, i_{n}$ be the resulting list
        \For{$i= i_1, i_2, \dots, i_n$} \label{alg_line:machine_scheduling_procedure_main_loop}
            \State $l^* \gets argmin_{l \in \{1, \dots, k\}} \max_{j \in J_l} (x_j + p_j)$ \label{alg_line:machine_scheduling_procedure_subset_selection}
            \Comment{$l^*$ is the index of the least loaded machine}
            \State $J_{l^*} = J_{l^*} \cup \{i\}$ \label{alg_line:machine_scheduling_procedure_job_addition}
        \EndFor
        \State \textbf{Return} $(J_1, J_2, \dots, J_k)$
    \end{algorithmic}
    The claim is that when $J_{l^*}$ is selected on line \ref{alg_line:machine_scheduling_procedure_subset_selection}, $x_i \geq \max_{j \in J_{l^*}} (x_j + p_j)$.
    Indeed, if $x_i < \max_{j \in J_{l^*}} (x_j + p_j)$, by definition of $l^*$, all machines are loaded and at least $k+1$ jobs are being processed simultaneously on the non-empty interval $[x_i, \max_{j \in J_{l^*}} (x_j + p_j))$, implying that $F(x) < \sum_{i \in J} p_i$, which is a contradiction.
    The inequality $x_i \geq \max_{j \in J_{l^*}} (x_j + p_j)$ being verified at each step of the procedure guarantees that the jobs in each $J_l$ do not overlap.
    A feasible solution for $Pk | \beta | C_{max}$ with makespan at most $M$ is therefore given by $(x, (J_1, J_2, \dots, J_k))$.
\end{proof}

The equivalence of Lemma \ref{lemma:general_reduction} provides a general reduction from a problem of the form $Pk | \beta | C_{max}$ to a problem of the form $Lk | \beta, C_{max} \leq M, c_i = 1| F$.
The $NP$-hardness (resp. strong $NP$-hardness) of the former thus implies the $NP$-hardness of the latter.

\begin{corollary} \label{corollary:grouped_complexity_result}
    The problems listed below are strongly $NP$-hard:
    \begin{itemize}
        \item[-] $L1 | r_i, d_i, c_i = 1 | F$
        \item[-] $L2 | chains, C_{max} \leq M, c_i = 1 | F$
        \item[-] $L | C_{max} \leq M, c_i = 1 | F$ 
        \item[-] $L | prec, C_{max} \leq M, c_i = 1, p_i = 1 | F$ 
    \end{itemize}
    Problem $L2 | C_{max} \leq M, c_i = 1 | F$ is $NP$-hard.
\end{corollary}

\begin{proof}
    Problems $1 | r_i, d_i | C_{max}$, $P2 | chains | C_{max}$, $P | | C_{max}$ and  $P | prec, p_i = 1 | C_{max}$ are known to be strongly $NP$-hard (see respectively \citep{Garey1977}, \citep{DuLeung1991}, \citep{Garey1978} and \citep{Ullman1975}).
    Problem $P2 | | C_{max}$ is known to be $NP$-hard \citep{Lenstra1977}.
\end{proof}

\begin{remark}
    While $P2 | prec | C_{max}$ and $P | prec, p_i = 1 | C_{max}$ are known to be strongly $NP$-hard, which can be used to prove the complexity of their leveling counterparts, the complexity of $P3 | prec, p_i = 1 | C_{max}$ remains unknown.
    It is then not possible to get a complexity result for $L3 | prec, C_max \leq M, c_i = 1, p_i = 1 | F$ based on Lemma \ref{lemma:general_reduction}.
\end{remark}

The $NP$-hardness of $L1 | r_i, d_i, c_i = 1 | F$ can actually be extended to the case of $L = 2$ quite naturally: 
an instance of $L1 | r_i, d_i, c_i = 1 | F$ reduces to an instance of $L2 | r_i, d_i, c_i = 1 | F$ with an additional job $i$ verifying $p_i = M$, $r_i = 0$ and $d_i = M$ since job $i$ uses one unit of resource from $0$ to $M$.
Hence the following corollary:

\begin{corollary} \label{corollary:L2_ri_di_unit_ci}
    $L2 | r_i, d_i, c_i = 1 | F$ is strongly $NP$-hard.
\end{corollary}

Another easy extension of Corollary \ref{corollary:grouped_complexity_result} can be made to prove that $L1 | chains | F$ is strongly $NP$-hard.
It requires to notice that an instance $I$ of $L1 | r_i, d_i, c_i = 1 | F$ reduces to an instance of $L1 | chains | F$ where each job in $I$ is given a predecessor and a successor that force it to be scheduled in its availability interval.
More precisely, for a job $i$ in $I$, two jobs $i'$ and $i''$ are added with $p_{i'} = r_i$, $p_{i''} = M - d_i$, $c_{i'} = c_{i''} = 0$ and precedence constraints $(i', i)$ and $(i, i'')$.
This gives the following corollary: 

\begin{corollary} \label{corollary:L1_chains}
    $L1 | chains | F$ is strongly $NP$-hard.
\end{corollary}

\section{Summary of results} \label{section:table}

Table \ref{tab:complexity} summarizes the complexity results that were obtained in this work.
Column headers are classical scheduling constraints: makespan deadline ($C_{max} \leq M$), precedence constraints ($prec$), possibly in the form of an $in \text{-} tree$ and release and due dates ($r_i, d_i$), which may be considered with or without preemption ($pmtn$).
Row headers are parameters restrictions: the value of $L$ may be restricted to one or two while unit values may be imposed for processing times and resource consumptions.
\\
The entries of the table include the complexity of the corresponding problem as well as the result proving it.
A natural question arising from this table is the complexity of problems for $L = 3$ or larger constant $L$, in particular $L3|prec,c_i=1, p_i=1|F$. 
This latter problem generalizes $P3|prec,p_i=1|C_{max}$, the complexity of which is a famous open problem in scheduling.

\begin{remark} \label{remark:trivial_result}
    Problem $L1 | C_{max} \leq M | F$ is solved to optimality by scheduling jobs with $c_i > 0$ one after another while $M$ is not exceeded and scheduling all remaining jobs at $M - p_i$.
    \\
    In the case of Problem $L | p_i = 1, c_i = 1, C_{max} \leq M | F$, all jobs are independent from each other and equivalent, with unit processing time and resource consumption.
    It is clear that scheduling one job per time step sequentially, going back to $0$ when $M$ is reached, solves the problem to optimality.
\end{remark}

\begin{table*}[h]
    \centering
    \small
    \caption{Complexity of resource leveling problems}
    \scalebox{0.75}{
    \begin{tabular}{|c|c||c|c|c|c|c|}
        \cline{3-7}
        \multicolumn{2}{c|}{}                       & \multirow{2}{*}{$C_{max} \leq M$}                                                                                                 & \multicolumn{2}{c|}{$prec, C_{max} \leq M$}                                                                                                                                           & \multicolumn{2}{c|}{$r_i, d_i$} \\ \cline{4-7}
        \multicolumn{2}{c|}{}                       &                                                                                                                                   & $prec = any$                                                                                              & $prec = in\text{-}tree$                                                          & $\text{no } pmtn$                                                                                 & $pmtn$                                                                                            \\ \cline{3-7} \noalign{\vskip\doublerulesep \vskip-\arrayrulewidth} \hline
        \multirow{2}{*}{$L1$}   & $\emptyset$       & \multirow{2}{*}{\makecell{in $P$\\(Remark \ref{remark:trivial_result})}}                                        & \multicolumn{2}{c|}{\makecell{str. $NP$-hard \\ (Cor. \ref{corollary:L1_chains})}}                                                                                                    & \multirow{2}{*}{\makecell{str. $NP$-hard \\ (Cor. \ref{corollary:grouped_complexity_result})}}    & \multirow{2}{*}{\makecell{in $P$ \\ (Prop. \ref{prop:premptive_interval_job_scheduling})}}        \\ \cline{2-2} \cline{4-5}
                                & $c_i > 0$         &                                                                                                                                   & \multicolumn{2}{c|}{\makecell{in $P$ \\ (Prop. \ref{prop:unit_resource_leveling_correctness})}}                                                                                       &                                                                                                   &                                                                                                   \\ \hline \hline
        \multirow{3}{*}{$L2$}   & $\emptyset$       & \multirow{2}{*}{\makecell{in $P_{pseudo}$ (Prop. \ref{prop:L2})\\$NP$-hard (Cor. \ref{corollary:grouped_complexity_result})}}     & \multicolumn{2}{c|}{\multirow{2}{*}{\makecell{str. $NP$-hard \\ (Cor. \ref{corollary:grouped_complexity_result})}}}                                                                   & \multirow{2}{*}{\makecell{str. $NP$-hard \\ (Cor. \ref{corollary:L2_ri_di_unit_ci})}}             & \multirow{2}{*}{\makecell{in $P$ \\ (Prop. \ref{prop:L2_preemptive_interval_job_scheduling})}}    \\ \cline{2-2}
                                & $c_i = 1$         &                                                                                                                                   & \multicolumn{2}{c|}{}                                                                                                                                                                 &                                                                                                   &                                                                                                   \\ \cline{2-7}
                                & $p_i = c_i = 1$   & \makecell{in $P$\\(Remark \ref{remark:trivial_result})}                                                        & \multicolumn{2}{c|}{\makecell{in $P$ \\ (Th. \ref{theorem:L2_UET_prec})}}                                                                                                             & \multicolumn{2}{c|}{\makecell{in $P$ \\ (Prop. \ref{prop:UET_interval_job_scheduling})}}                                                                                                              \\ \hline \hline
        \multirow{3}{*}{$L$}    & $\emptyset$       & \multicolumn{4}{c|}{\multirow{2}{*}{\makecell{str. $NP$-hard \\ (Cor. \ref{corollary:grouped_complexity_result})}}}                                                                                                                                                                                                                                                                                                           & \makecell{str. $NP$-hard \\ \footnotesize{\citep{Gyorgyi2020}}}                                                 \\ \cline{2-2} \cline{7-7}
                                & $c_i = 1$         & \multicolumn{4}{c|}{}                                                                                                                                                                                                                                                                                                                                                                                                         & \makecell{in $P$ \\ (Prop. \ref{prop:premptive_interval_job_scheduling})}                         \\ \cline{2-7}
                                & $p_i = c_i = 1$   & \makecell{in $P$\\(Remark \ref{remark:trivial_result})}                                                         & \makecell{str. $NP$-hard \\ (Cor. \ref{corollary:grouped_complexity_result})}    ,                       & \makecell{in $P$ \\ (Prop. \ref{prop:optimality_adapted_hu_algorithm})}    & \multicolumn{2}{c|}{\makecell{in $P$ \\ (Prop. \ref{prop:UET_interval_job_scheduling})}}                                                                                                              \\ \hline         
    \end{tabular}
    }
    \label{tab:complexity}
\end{table*}

\section{Conclusion} \label{section:conclusion}

A polynomial-time algorithm solving a resource leveling counterpart of a well-known two-processor scheduling problem is proposed in this paper as a main result.
As complementary results, various related resource leveling problems are studied, some of which are shown to be solvable in polynomial time as well.
Close links are highlighted between classical machine scheduling resource leveling.
Both solving methods and $NP$-hardness results are mirrored from one field to the other.
While the translation of negative complexity results is rather straightforward, the main contribution of this work lies on the design of algorithms.
Besides the core problem, an interesting result in that regard is the adaptation of Hu's algorithm to solve the case of unit processing times with unit resource consumptions and an in-tree precedence graph.

Although polynomial and pseudo-polynomial algorithms are given, the problem variants that they can solve remain special cases.
An interesting perspective would be to consider more advanced constraints and less restrictive parameters, to meet practical requirements.
In particular, some problem variants are polynomial for unit processing times but become strongly $NP$-hard in the general case.
A possible extension of this work would then be to find good approximation algorithms for some of the more realistic resource leveling problems.
\\
The objective function that was chosen is certainly relevant to model resource overload costs, yet it does not prevent peaks in resource use, which are often ill-advised in practice.
Studying the resource investment objective, which minimizes the highest amount of resource consumption, could therefore be another perspective for future research.

\begin{table*}[h]
    \centering
    \small
    \caption{Table of notations}
    \begin{tabular}{l l}
        \multicolumn{2}{l}{General}\\ 
        $J$                                 & Set of jobs\\
        $p_i$                               & Processing time of job $i$\\
        $c_i$                               & Resource consumption of job $i$\\
        $G$                                 & Precedence graph\\
        $\mathcal{A}$                       & Set of arcs in $G$\\
        $x$                                 & Schedule\\
        $x_i$                               & Starting time of job $i$ in schedule $x$\\
        $L$                                 & Resource level\\
        $M$                                 & Deadline on makespan\\
        $F$                                 & Objective function\\
        $P$                                 & Critical path in $G$\\
        \multicolumn{2}{l}{Scheduling constraints}\\ 
        $C_{max}$                           & Makespan\\
        $prec$                              & Precedence constraints\\
        $in \text{-} tree$                  & In-tree precedence graph\\
        $pmtn$                              & Preemption allowed\\
        $r_i$                               & Release date of job $i$\\
        $d_i$                               & Due date of job $i$\\
        \multicolumn{2}{l}{Auxiliary graphs}\\
        $\widetilde{G}$                     & Independence graph\\
        $E$                                 & Set of edges in $\widetilde{G}$\\
        $m^*$                               & Size of a maximum matching in $\widetilde{G}$\\
        $\widetilde{G}_P$                   & Bipartite independence graph for critical path $P$\\
        $E_P$                               & Set of edges in $\widetilde{G}_P$\\
        $m^*_P$                             & Size of a maximum matching in $\widetilde{G}_P$\\
    \end{tabular}
    \label{tab:notations}
\end{table*}

\bibliography{references}

\end{document}